\documentclass{article}[11pt]
\newcommand{\Description}[1]{}

\usepackage{amsmath,amsthm,amssymb,nicefrac}
\usepackage{microtype}

\usepackage{fullpage}
\usepackage{setspace}
\usepackage[breaklinks]{hyperref}
\hypersetup{colorlinks=true,%
            citebordercolor={.6 .6 .6},linkbordercolor={.6 .6 .6},%
citecolor=blue,urlcolor=black,linkcolor=blue,pagecolor=black}
\usepackage[compact]{titlesec}
\newcommand{\removelatexerror}{}

\usepackage[nameinlink]{cleveref}
\Crefname{algocf}{Algorithm}{Algorithms}
\crefname{algocfline}{line}{lines}

\usepackage{epsfig}
\usepackage{amsthm,mathrsfs}
\usepackage{xspace}
\usepackage{soul}
\usepackage{latexsym}

\usepackage[dvipsnames]{xcolor}
\definecolor{DarkGray}{rgb}{0.66, 0.66, 0.66}
\definecolor{DarkPowderBlue}{rgb}{0.0, 0.2, 0.6}
\definecolor{fluorescentyellow}{rgb}{0.8, 1.0, 0.0}

\usepackage[ruled,vlined,linesnumbered]{algorithm2e}
\SetKwComment{Comment}{\footnotesize$\triangleright$\ }{}

\SetCommentSty{mycommfont}

\usepackage{thmtools,thm-restate}

\usepackage{fullpage}
\makeatletter
\allowdisplaybreaks
\makeatother

\newcounter{note}[section]
\renewcommand{\thenote}{\thesection.\arabic{note}}
\newcommand{\agnote}[1]{\refstepcounter{note}$\ll${\bf Anupam~\thenote:}
  {\sf \color{red} #1}$\gg$\marginpar{\tiny\bf AG~\thenote}}

\sethlcolor{fluorescentyellow}

\newcommand{\initOneLiners}{%
    \setlength{\itemsep}{0pt}
    \setlength{\parsep }{0pt}
    \setlength{\topsep }{0pt}
}
\newenvironment{OneLiners}[1][\ensuremath{\bullet}]
    {\begin{list}
        {#1}
        {\initOneLiners}}
    {\end{list}}

\newtheorem{theorem}{Theorem}[section]
\newtheorem{lemma}[theorem]{Lemma}
\newtheorem{claim}[theorem]{Claim}
\newtheorem{corollary}[theorem]{Corollary}
\newtheorem{definition}[theorem]{Definition}

\newcommand{\eat}[1]{}

\newcommand{\T}{{\mathcal{T}}}
\newcommand{\fT}{{\mathfrak{T}}}

\newcommand{\eps}{\varepsilon}
\newcommand{\sse}{\subseteq}

\newcommand{\I}{{\mathcal{I}}}

\newcommand{\opt}{{\textsf{opt}}}

\newcommand{\newS}{{B^\star}}

\newcommand{\oldS}{{A^\star}}
\newcommand{\cT}{{\cal T}}

\newcommand{\pagecover}{{\textsf{SolveDext}}\xspace}
\newcommand{\pagecoverp}{{\textsf{SolveDextP}}\xspace}
\newcommand{\specialpagecover}{{\textsf{NonNestDext}}\xspace}
\newcommand{\specialpagecoverp}{{\textsf{NonNestDextP}}\xspace}

\newcommand{\othercoverp}{{\textsf{SolveRextP}}\xspace}

\newcommand{\intervalcover}{{\textsf{TiledIC}}\xspace}
\newcommand{\exintervalcover}{{\textsf{TiledICEx}}\xspace}

\newcommand{\Z}{\mathbb{Z}}
\newcommand{\R}{\mathbb{R}}
\newcommand{\page}{\textsf{page}}
\newcommand{\cost}{\textsf{cost}}
\newcommand{\wPwTw}{{\textsf{PageTW}}\xspace}
\newcommand{\wPwTwP}{{\textsf{PageTWPenalties}}\xspace}
\newcommand{\wPwD}{{\textsf{PageD}}\xspace}

\newcommand{\dext}[1]{{\mathsf{Dext}}(#1)}
\newcommand{\rext}[1]{{\mathsf{Rext}}(#1)}
\newcommand{\astar}{A^\star}
\newcommand{\pstar}{p^\star}
\newcommand{\pdag}{p^\dagger}
\newcommand{\Udag}{U^\dagger}
\newcommand{\rt}{{\tt rt}}
\newcommand{\lt}{{\tt lt}}
\newcommand{\VC}{{\textsf{Vertex Cover}}\xspace}

\renewcommand{\emptyset}{\varnothing}

\newcommand{\tz}{\tilde{z}}

\newcommand{\calC}{\mathcal{C}}
\newcommand{\calD}{\mathcal{D}}
\newcommand{\calS}{\mathcal{S}}

\newcommand{\calN}{\mathcal{N}}

\newcommand{\calI}{\mathcal{I}}
\newcommand{\cK}{\mathcal{K}}

\newcommand{\tx}{{\tilde x}}
\newcommand{\bx}{{\bar x}}
\newcommand{\ty}{{\tilde y}}
\newcommand{\by}{{\bar y}}
\newcommand{\penaltyfn}{F}

\begin{document}

\title{Caching with Time Windows and Delays}

\author{
{Anupam Gupta\thanks{Computer Science Department, Carnegie Mellon University, Pittsburgh, PA. Email: {\tt anupamg@cs.cmu.edu}.}}
\and
{Amit Kumar\thanks{Department of Computer Science and Engineering, IIT Delhi, New Delhi, India. Email: {\tt amitk@cse.iitd.ac.in}.}}
\and
{Debmalya Panigrahi\thanks{Department of Computer Science, Duke University, Durham, NC. Email: {\tt debmalya@cs.duke.edu}.}}
}

\date{}

\maketitle

\begin{abstract}

  \medskip
We consider two generalizations of the classical weighted 
paging problem that incorporate the notion of delayed service
of page requests. The first is the {\em (weighted) Paging with Time Windows}
(\wPwTw) problem, which is like the classical weighted
paging problem except that each page request only needs to be served before 
a given deadline. This problem arises in many practical 
applications of online caching, such as the ``deadline'' I/O
scheduler in the Linux kernel and video-on-demand streaming.
The second, and more general, problem is the 
{\em (weighted) Paging with Delay} (\wPwD) problem, 
where the delay in serving a page request results in a penalty
being assessed to the objective. This problem generalizes the caching 
problem to allow delayed service, a line of work that has 
recently gained traction in online algorithms 
(e.g., Emek {\em et al.} STOC '16, 
Azar {\em et al.} STOC '17, Azar and Touitou FOCS '19).

We give $O(\log k\log n)$-competitive algorithms for both the
\wPwTw and \wPwD problems on $n$ pages with a cache of size $k$. This
significantly improves on the previous best bounds of $O(k)$ 
for both problems (Azar {\em  et al.} STOC '17).
We also consider the {\em offline} \wPwTw and \wPwD problems, for which we give
an $O(1)$ approximation algorithms and prove APX-hardness.  These are
the first results for the offline problems; even NP-hardness was not
known before our work.
At the heart
of our algorithms is a novel ``hitting-set'' LP relaxation 
of the \wPwTw problem that overcomes the $\Omega(k)$ integrality
gap of the natural LP for the problem. To the best of our knowledge,
this is the first example of an LP-based algorithm for an online
algorithm with delays/deadlines.
\end{abstract}

\section{Introduction}
\label{sec:introduction}

In the caching/paging problem,
page requests from a universe of $n$ pages arrive over time. They 
have to be served by swapping pages in and out of a cache that can hold
only $k < n$ pages at a time. In weighted paging, each page $p$ has a 
weight $w_p$, and the goal is to minimize the sum of weights of evicted pages. 
In this paper we consider situations where page requests do not need to be 
served immediately, but can be delayed for some time. For instance, 
in mixed-workload environments such as those arising in cloud computing 
or operating systems, requests from 
time-sensitive applications (such as interactive ones) 
have short deadlines, but batch processes can tolerate longer wait times.
(Indeed, the ``deadline'' I/O scheduler in the Linux kernel is precisely
for this purpose, although the way it currently handles deadlines is not
very sophisticated~\cite{LinKernel}.) A different application arises in
network streaming, e.g., in video-on-demand, where a server needs to
cache segments appearing in multiple video streams (see, e.g.,
\cite{ClaeysBDVLT, Schepper}). Depending on when these segments are
required, various streams set different deadlines for each of these 
segments. In all these applications, the key feature is that individual
page requests can be delayed, but only until a given deadline. Specifically, 
the request $r_t = (p, d)$ at time $t$ for a page $p$ includes a
deadline $d$, and the algorithm must ensure that the page is in the
cache at some time in the interval $[t,d]$. 
We call this the {\em (weighted) Paging with Time Windows}
(\wPwTw) problem; if the deadline is the same as the time of the
request, we get back the weighted paging problem. 

A more general setting is one where the page requests do not have 
specific deadlines, but the algorithm incurs a cost that is monotonically
non-decreasing with the delay in serving individual requests. This is related
to the recent line of work in online algorithms with delay, where
problems such as online matching~\cite{EmekKW16,AzarCK17,AshlagiACCGKMWW17,AzarF20}
and online network design~\cite{AzarT19,AzarT20}
have been considered. In particular, our work relates to the ``online service
with delays'' problem~\cite{AzarGGP17,BienkowskiKS18,AzarT19}, and can be 
interpreted as a generalization of this problem to $k$ servers but for the 
special case of a star metric. In our problem, 
each request is specified by a triple
$(p,t, \penaltyfn)$, where $p$ is the requested page, $t$ is the time
at which this request is made, and
$\penaltyfn: \{t, t+1, \ldots, \} \to \R_{\geq 0}$ denotes the
non-decreasing loss function associated with it. 
The objective is to minimize
the sum of two quantities: the sum of weights of pages evicted
from the cache (the usual objective in weighted paging)
and {\em the total delay losses incurred over all the individual
page requests}. We call this the {\em (weighted) Paging with Delay}
(\wPwD) problem. Note that \wPwTw is a special case of this problem
where the delay loss is $0$ till the deadline, and $\infty$ thereafter.

\begin{theorem}[Main Results: Online Algorithm]
  \label{thm:main-on}
  There is an $O(\log k \log n)$-competitive randomized algorithm for
  the \wPwD problem in the online setting, where $n$ is the number of
  pages and $k$ is the size of the cache. As a consequence, there
  is also an $O(\log k \log n)$-competitive randomized algorithm for
  the special case of the \wPwTw problem in the online setting.
\end{theorem}

Previously, an $O(k)$-competitive deterministic algorithm was given by
Azar {\em et al.}~\cite{AzarGGP17} for both problems.  
\wPwD and \wPwTw inherit an
$\Omega(\log k)$-competitiveness lower bound from the classical paging
problem; closing the gap between our upper bound and this lower bound
remains open.

While we stated the above theorem for the more general \wPwD problem,
and derived the bound for \wPwTw as a corollary, we will actually
prove the theorem for the special case of the \wPwTw problem first, 
and then show that we can reduce the \wPwD problem to the \wPwTw 
problem. %
More precisely,
we extend our \wPwTw algorithm to a generalization that we call
the \wPwTwP problem, where every page request has a non-negative 
penalty that the algorithm can choose to incur instead of satisfying
the request. Then, we give a reduction from the \wPwD problem
to the \wPwTwP problem in \S\ref{sec:penalties} without changing the 
objective; %
moreover, this reduction can be performed online.
So, the rest of 
this section, and much of the subsequent sections, focus on the \wPwTw 
and \wPwTwP problems. 

At the heart of our algorithm is a novel ``hitting-set'' LP relaxation
of the \wPwTw problem that overcomes the $\Omega(k)$ integrality gap
of the natural LP relaxation for this problem (see \Cref{sec:bad-ex}). From a theoretical
perspective, the \wPwTw problem is in the category of online
optimization problems with delays/deadlines that has attracted
significant interest recently (e.g.,
\cite{EmekKW16,AzarGGP17,AzarCK17,BuchbinderFNT17,AzarT19,BienkowskiBBCDF16}).
To the best of our knowledge, our work is the first example of an
LP-based algorithm in this line of research. Given the great success
of LP-based techniques in online algorithms in general, we hope that
our work spurs further progress in this area.

We also study the {\em offline} versions of the \wPwTw and \wPwD problems, 
where the request sequence is given up-front. Here, the first question is
tractability: since weighted paging is solvable in polynomial time
offline, it is conceivable that so are \wPwTw and \wPwD. 
We show that the \wPwTw problem (and therefore, by generalization,
the \wPwD problem) is %
APX-hard. We complement this lower bound with an $O(1)$-approximation
for the offline \wPwTw and \wPwTwP problems, which again by our reduction from
the \wPwD problem to the \wPwTwP problem, implies an $O(1)$-approximation
for the offline \wPwD problem.

\begin{theorem}[Main Results: Offline Algorithm]
  \label{thm:main-off}
  The \wPwTw problem is NP-hard (and APX-hard) even when the cache
  size $k=1$, and the pages have unit weight. As a consequence, the 
  \wPwD problem is also NP-hard (and APX-hard) under these restrictions.
  Moreover, there are $O(1)$-approximation deterministic algorithms
  for the \wPwTw and \wPwD problems, based on rounding a linear program
  to show a constant integrality gap.
\end{theorem}

\subsection{Our Techniques}

\begin{figure}
    \centering
    \includegraphics[width=4in]{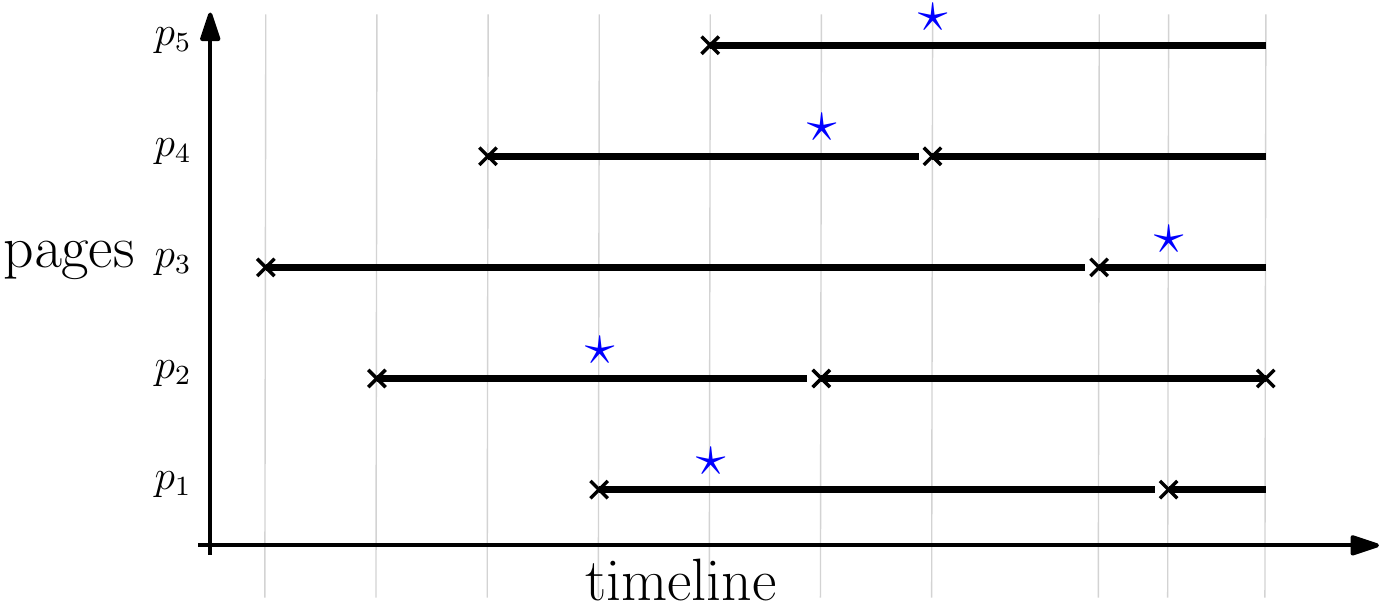}
    \caption{\small\emph{A two-dimensional view of page requests and evictions. 
     The crosses represent page requests and stars represent page evictions.
     This illustration is for a cache of size $3$. }}
    \label{fig:stars}
\end{figure}

The weigh\-ted paging problem has an ``interval covering'' IP formulation~\cite{BansalBN12, Young91}:
$$\min\left\{\sum_{p, j} w_p x_{p,j}: \sum_{p\not= p_t} x_{p, j(p, t)} \geq n-k ~\forall t, x_{p, j}\in \{0, 1\}~ \forall p, j\right\}.$$
For every page $p$, define an
interval starting at each request for it, and ending just before the
next request. Because of the request, this page $p$ must be present in
the cache at the start of each such interval, but may be evicted at
some subsequent point: the IP variable $x_{p, j} \in \{0,1\}$ %
indicates if a page is evicted before its next request. 
While this IP does not explicitly indicate {\em when} a page is evicted, 
any online algorithm solving it must raise a variable $x_{p, j}$ from $0$ 
to $1$ at a specific time between the $j$th and $(j+1)$st request for 
page $p$. We visualize this using a $2$-dimensional
picture indexed by the pages and time, recording the eviction of 
page $p$ at time $t$ by putting a star at location $(p,t)$
(see Fig.~\ref{fig:stars}).
In classical paging, the intervals for any page partition 
its row into disjoint, tightly-fitting segments. The capacity
constraint of the cache forces the following 
property: of the $n$ intervals (for different pages) containing  
time $t$ (these are indexed $j(p, t)$ for page $p$), 
at least $n-k$ contain a star at some time $\leq t$.
In other words, at least $n-k$ pages must have been evicted from the 
cache since their last request.

The situation is more complex in \wPwTw. Previously, it sufficed 
to record page evictions, because page insertions 
were entirely dictated by the requests: whenever a page 
is requested, it must be inserted in the cache if it were evicted 
after its previous request. So, for an insertion to be feasible, it suffices to just ensure that 
sufficiently many pages are evicted since their previous 
request. In \wPwTw, however, 
page requests can be fulfilled at a later time, so 
evictions alone do not completely describe the state of the cache. 
One option is to explicitly encode page insertions via IP variables,
but then we need packing constraints on these variables to 
enforce the size of the cache. Handling such 
packing constraints in online IPs seems beyond the scope of current
techniques in online algorithms. Another idea is to  
reduce the ambiguity of when pages are inserted in the cache, e.g.,
by enforcing that all page insertions are done the end of their  
request intervals (if the page is not in the cache at the beginning of 
the interval). This would be a useful property, because the state of the 
cache could then be completely described by variables for page evictions. 
This property, however, is false: forcing a page
request to be satisfied at the start/end of its request interval
can be much costlier than doing it somewhere in the middle.
E.g., if a heavy page is being evicted, we should serve some outstanding light requests while there 
is an empty slot in the cache (see 
\Cref{sec:endpoints-example}).

\paragraph{The Hitting Set IP Relaxation.}
To overcome these challenges, we first re-interpret the 
interval covering IP for classical paging.
We again 
use $x_{p, t}$ variables (saying page $p$ is evicted at time $t$).
The cache-size constraint at any time $t'$ 
insists that at least $n-k$ pages are evicted at 
times $\leq t'$ 
since their last request. To implement this, we define an interval for each page $p$ starting at 
the last request for $p$ and ending at $t'$, and write a covering
constraint saying at least $n-k$ of 
these intervals have a star within them
(i.e., $x_{p, t} = 1$ for times $t$ within 
these intervals). Note: there is nothing special about the {\em last}
request for a page before $t'$---we could 
have written these constraints for every choice of 
request of every page before $t'$. 
The additional constraints would be redundant 
given the one containing the last requests, and 
would unnecessarily lead to an exponential-sized IP.

In the \wPwTw problem, however, the request intervals
for a page might overlap, or may even be nested, so it is easier to write constraints
for {\em every} request, rather than to identify
some (non-canonical) {\em last} request before time $t'$.
Extending the previous intuition, we define the following intervals for time 
$t'$: corresponding to a request interval $I = (s(I), e(I))$ for a page
with $s(I) < t'$, there is a  constraint interval $(s(I), \max(e(I), t'))$.
Note that if the request interval extends beyond $t'$, i.e., 
$e(I) > t'$, then we must extend the constraint 
interval \textsf{rightwards} to $e(I)$, since the page might be served 
{\em after} $t'$.
Now, we enforce the same constraint as earlier: for 
any choice of such constraint intervals, one for each distinct page,
at least $n-k$ must have a star in them. We call
these constraint intervals the {\em right extensions} of their respective 
request intervals at time $t'$ (see Fig.~\ref{fig:exts} for an example).

\begin{figure}
    \centering
    \includegraphics[width=5in]{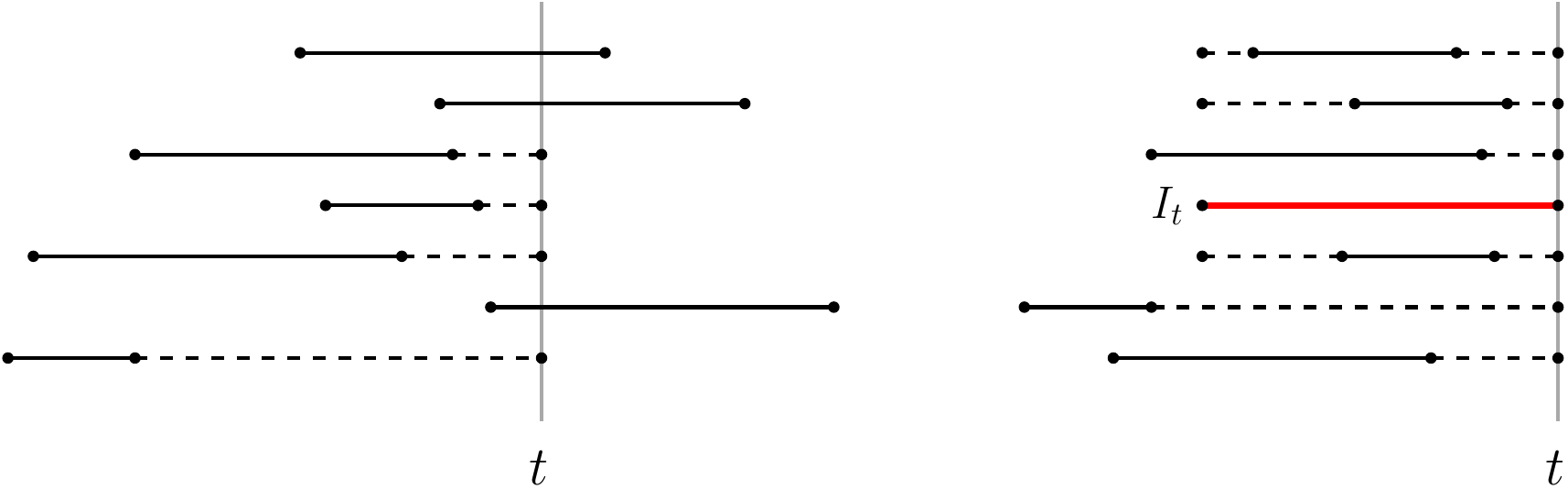}
    \caption{\emph{The left figure illustrates right extensions 
        $\rext{I,t}$ for request
        intervals $I$ shown by solid lines. The dotted lines show how
        these intervals are extended. The figure on the right shows
        double extensions
        $\dext{I,t}$, with the critical interval $I_t$ in red.}}
    \label{fig:exts}
\end{figure}

In classical paging, these %
constraints are valid even if we exclude the 
page currently requested at time $t'$. In other words, 
of the remaining $n-1$ pages, the constraint ensures that at 
least $n-k$ have been evicted ensuring a cache slot for the 
currently requested page. %
All feasible solutions satisfy this constraint since the requested page {\em must}
be in the cache at time $t'$.
This stronger constraint, however, does not hold for \wPwTw. If 
we write the above constraints for $n-1$ pages, then it 
would reserve a cache slot for the remaining page at the current 
time, thereby excluding feasible solutions that do not satisfy 
this property. Conversely, the (weaker) constraints
summing over all $n$ (and not $n-1$) pages 
are not sufficient: they do not reserve 
a cache slot for a requested page at any time during the 
request interval. 

So we need a new set of constraints.
These reserve a cache slot for a page $p$ within 
each request interval $I_t = (s_p, t_p)$ for it. 
Let us exclude this page $p$
and choose a request $I = (s(I), e(I))$, where 
$e(I) \leq t_p$, for each of the remaining $n-1$ pages. 
For each such request (say for a page $p'$), consider a different extended 
constraint interval $(\min(s_p, s(I)), t_p)$. Now we are guaranteed
that in any feasible solution, one of two things happens:
either page $p'$ resides in the cache for the entire extended
interval $(\min(s_p, s(I)), t_p)$ and therefore also for the sub-interval
$(s_p, t_p)$, or it is ``hit'' (inserted or evicted) during
the constraint interval. Since page $p$ must be served in 
its request interval
$(s_p, t_p)$, at most $k-1$ pages can be resident in the cache
during $(s_p, t_p)$, i.e, at least $n-k$ of the $n-1$ pages
are hit during these extended constraint intervals. We call these 
extended intervals {\em double extensions} of their  
request intervals for time $t_p$ (again, see Fig.~\ref{fig:exts}).
Our ``hitting set'' IP comprises these two sets
of requests, for right extensions and double extensions.
We give details of this formulation in \S\ref{sec:IP}. 

\paragraph{Solving the Hitting Set IP Online.}
Loosely, we extend ideas from 
Bansal~{\em et al.}~\cite{BansalBN12} for solving the 
weighted paging IP online to our hitting set IP. There
are some challenges, however. Firstly, the 
hitting set IP
is of exponential size, since we wrote covering constraints
for every choice of request interval for every page. 
Secondly, unlike in weighted paging, there are two sets of 
constraints, one on $n-1$ pages and the other on $n$ pages;  
the weighted paging IP only has the first set. Thirdly, the decision variables are for
$(p, t)$ pairs, and do not uniquely correspond to 
constraint intervals. Nevertheless, we show that, as long as the request intervals for the pages are ``non-nested'', the 
techniques of \cite{BansalBN12} can be adapted to 
our hitting set IP. When the request intervals are nested, we solve
the problem on two carefully selected subsets of the input where the
request intervals are non-nested; then we show, somewhat surprisingly, that the combined solution satisfies the general instance. The competitive 
ratio of this algorithm is $O(\log k)$, asymptotically the same as
weighted paging. This algorithm appears in~\S\ref{sec:first-lp}.

\paragraph{Converting IP Solution to Cache Schedule Online.}
As described earlier, the IP solution only gives us a set of stars,
indicating ``hits'' for each page where each hit might either
represent insertion or eviction of the page from the cache.  Moreover,
the IP solution does not necessarily give {\em all} the insertions and
evictions. E.g., in the case of instantaneous request intervals
representing classical weighted paging, our IP is identical to the
standard interval covering IP and only gives page evictions. Indeed,
the bulk of our technical work is in converting a
feasible IP solution to an actual cache schedule that satisfies all
requests. This is further complicated by the fact that this
translation has to be done online.

The main difficulty is the following: when a
request interval $I$ for a page $p$ arrives, we don't know how long to
wait before serving it. For instance, suppose the IP has a 
``hit'' for a heavy page. The example in 
\Cref{sec:endpoints-example}
shows that we must use this opportunity to serve requests for 
light pages that are currently waiting. But, which pages should 
we serve? Suppose we serve a page $p$ at some time $t \in I$ by
loading $p$ in the cache, and evict it soon after to serve 
other light pages. Now if another interval $I'$ for $p$ arrives after
$t$ and $I'$ overlaps with (or is even nested in) $I$, 
it is clear that we should have waited to load $p$
till $I'$ arrives. In the offline case, we can use a
reverse-delete step where we undo such mistakes. But, in the online
setting, we must find a careful balance between waiting ``long
enough'' and servicing outstanding requests. Specifically, we build a
tree structure over the request intervals (which may not be laminar
in general), and use the structural properties to argue that our algorithm can find a balance between these two competing goals. 
The online conversion 
algorithm appears in \S\ref{sec:solving-wpwtw-using}. 

While we cannot show that our
algorithm achieves the ultimate goal of being $O(\log k)$-competitive,
we do not know any worse gaps for our approach. Indeed, the fact that
the integrality gap of the hitting set formulation is constant, as 
evidenced by our offline solution, gives us hope that the ideas here will
lead to further improvements.

\subsection{Related Work}
\label{sec:related-work}

Azar et al.~\cite{AzarGGP17} study the online service problem with
delays, where a single server services requests in a metric
space. Each request has an associated monotone delay function that
gives the cost of serving requests at each time after its arrival. The
server pays for the total movement plus delay costs. They give an
$O(h^3)$-competitive algorithm for HSTs of height $h$. They extend the
result to $k$ servers at a loss of a factor of $k$, which gives an
$O(k)$-competitiveness for \wPwTw. (Progress
on related problems appears in~\cite{AzarT19}.)  A related problem
is \emph{online multilevel aggregation}~\cite{BienkowskiBBCDF16} where
a single server sits at the root of a tree, requests arrive at the
leaves, and the server occasionally goes to service some subset of
requests and returns to the root. The cost is again the sum of
movement and delay costs. Buchbinder et al.\ gave an
$O(h)$-competitive algorithm for $h$-level
HSTs~\cite{BuchbinderFNT17}, improving on~\cite{BienkowskiBBCDF16};
the model itself combines elements of TCP
acknowledgment~\cite{KarlinKR03} and online joint
replenishment~\cite{BuchbinderKLMS13}. Online problems with delays
were first proposed by Emek et al.~\cite{EmekKW16} for online
matching; see~\cite{AzarCK17,AshlagiACCGKMWW17} for other
work. %

In the classical paging/caching problem with instantaneous requests,
each interval is of length zero and must be satisfied immediately.
Belady's offline algorithm (Farthest in Future) is optimal for the
number of evictions~\cite{Belay66}; in contrast, the offline \wPwTw
problem is APX-hard. We know deterministic $k$-competitive and
randomized $O(\ln k)$-competitive algorithms; both are
optimal~\cite{SleatorT85,FiatKLMSY91}.  \emph{Weighted paging} is
equivalent to the $k$-server problem on a weighted star, so
deterministic $k$-competitiveness follows from the algorithm
$k$-server on
trees~\cite{ChrobakKPV91}.  %
Bansal et al.~\cite{BansalBN12} gave a randomized
$O(\ln k)$-competitive algorithm for weighted paging, illustrating the
power of the primal-dual technique for these
problems. They used an interval covering IP give
by~\cite{Bar-NoyBFNS01, CohenK99}, which we extend in our work.

\paragraph{Paper Outline.}
In Section~\ref{sec:IP},
we describe the new~\ref{eq:IP} formulation for \wPwTw. In
Section~\ref{sec:solving-wpwtw-using}, we show how a solution to
this~\ref{eq:IP} can be used to generate a caching schedule online. 
We give the corresponding offline algorithm in \Cref{sec:overlap}.
We show how to (approximately) solve the IP, both offline and online, in
\Cref{sec:first-lp}.
The reduction from the \wPwD problem to the \wPwTw problem
that changes the objective by at most an $O(1)$-factor 
appears in \Cref{sec:penalties}.
We prove APX-hardness of \wPwTw
in \Cref{sec:np-hard}, and give some illustrative examples in
\Cref{sec:exs}.

\section{The IP Relaxations for \wPwTw and \wPwTwP}
\label{sec:IP}

There is a universe of $n$ pages, and the cache can hold $k$ pages at
any time. Each page $p$ incurs a cost when we evict it from the cache,
which is denoted by its \emph{weight} $w(p) \geq 0$.  In the {\em
  (weighted) Paging with Time Windows} (\wPwTw) problem, each request
specifies a page $p$ and an interval $I = [s(I),
t(I)]$: %
the page $p$ must be in the cache at some time during this interval
$I$. Since the only times of interest in the problem are the start and
end times of intervals, we assume without loss of generality (wlog)
that $s(I), e(I)\in \Z$,
so the interval
$I := [s(I), \ldots, e(I)]$.  %
Note that in the traditional paging problem, each interval $I$
contains a single timestep, i.e., $I = \{t\}$ for some $t$.
In the online setting, a request comprising the identity of
the page and the end time of the interval $e(I)$ 
(i.e., the {\em deadline}) is revealed at its start time $s(I)$.
This is known as the {\em clairvoyant} setting in the 
literature; strong lower bounds are known for the  
 non-clairvoyant setting where the deadline is 
only revealed at time $t(I)$~\cite{AzarGGP17}.

We write a ``hitting set'' \emph{integer} programming relaxation for
this problem: this IP does not capture the \wPwTw problem exactly, but
we show that (a)~it contains only valid constraints, and hence provides
a lower bound on the optimal cost, (b)~it can be solved approximately in
polynomial time, and (c)~the ``relaxation'' gap is small, i.e., a
solution to this IP can be used to obtain a feasible solution to the
original \wPwTw problem.

The IP has Boolean variables $x_{pt}$ for each page-time pair $(p,t)$,
with this variable being set if the page $p$ is ``hit'' at time $t$:
it is either brought into or evicted from the cache at time $t$.  
We assume that for each time $t \in \Z$, there is exactly one request
interval $I$ having $e(I) = t$; this incurs no loss of generality, since
we can remove timesteps with no deadlines, and split times with multiple
intervals ending at it. Hence each request interval $I$ corresponds to a
unique page $\page(I) \in [n]$. For time $t$, let $I_t$ and $p_t$ be the
unique interval ending at time $t$, and its corresponding page; we call
these the \emph{critical interval and page} for time $t$.

As described in the introduction, we use two sets of extensions for 
request intervals to define these constraints: 
\begin{gather}
 \text{if $s(I) \leq t \quad \implies\quad$ right extension of $I$~~~}
 \rext{I,t} := [s(I), \ldots, \max(t, e(I))]. \\
 \text{if $e(I) \leq t \quad \implies\quad$ double extension of $I$~~~}   \dext{I,t} := [\min(s(I_t), s(I)), \ldots, t].
\end{gather}
(See Figure~\ref{fig:exts}.)
The ``Hitting Set IP'' below has variables $x_{pt} \in \{0,1\}$:
\begin{alignat}{2}
  \min \sum_{p,t} w(p) \, & x_{p,t} \tag{IP} \label{eq:IP}\\
  \sum_{I \in \calC} \sum_{t' \in \rext{I,t}} x_{\page(I),t'} &\geq 1
  & \quad &\forall t \, \forall \calC \text{ with $k+1$ requests for
    distinct pages starting
    before $t$}  \tag{R1} \label{eq:1} \\
  \sum_{I \in \calC} \sum_{t' \in \dext{I,t}} x_{\page(I),t'} &\geq 1
  & \quad &\forall t \, \forall \calC \text{ with $k$ requests for distinct
    pages  (excluding $p_t$)  ending before $t$}  \tag{D1} \label{eq:2} 
\end{alignat}
We now show that these sets of constraints are valid:
\begin{restatable}{claim}{IPValid}
\label{cl:valid}
  The constraints~(\ref{eq:1}) and~(\ref{eq:2}) are valid for any  solution to \wPwTw.
\end{restatable}
 \begin{proof}
  Fix a solution for \wPwTw. Set $x_{p.t}$ to 1 if this page $p$ is evicted at time $t$ or loaded into the cache at time $t$ in this solution.  Consider the
  constraint~(\ref{eq:1}) for a collection $\calC$ and time $t$. At time $t$, one
  of the $k+1$ pages corresponding to $\calC$ is not in the cache---let the
  corresponding request interval be $I \in \calC$ for page $p =
  \page(I)$. Two cases arise: in the  solution the page $p$ is in
  the cache either (a)~at some time during $[s(I),t)$, or (ii)~at some
  time during $[t, e(I)]$. (The latter case arises only if
  $t \leq e(I)$.) In the first case, $p$ must have been evicted during
  $[s(I), t]$, whereas in the second case it must have been brought
  into the cache during $[t, e(I)]$. In either case the $x_{p,t'}$
  variables sum to at least $1$ over the right-extended interval for $I$ with
  respect to $t$.

  Now consider the constraint~(\ref{eq:2}) for a collection $\calC$ and time $t$,
  where the critical request at time $t$ is $I_t$ with $p_t :=
  \page(I_t)$. In the solution, let page $p_t$ be in the cache
  at some time $\tau \in I_t$. At this time, at least one of the $k$ pages
  corresponding to the intervals in $\calC$ is not in the cache; say this
  interval is $I$ for page $p := \page(I)$. Again two cases arise: in
  the optimal solution this page is in the cache either (i)~at some time
  during $[s(I), \tau]$ (where this case arises only if $s(I) \leq
  \tau$), or (ii)~at some time $t$ during $[\tau, e(I)]$ (again, this case arises only if $\tau \leq e(I)$ -- note that at least one of these two cases must happen). If the former
  case happens, then $p$ must have been evicted during $[s(I), \tau]$,
  whereas if the second case happens, then $p$ must have been brought in
  the cache during $[\tau, e(I)]$. Since $\tau \in I_t \in
  [s(I_t),e(I_t) = t]$, in both cases the $x_{p,t'}$ variables sum to
at least   $1$ for the doubly-extended interval for $I$ with respect to $t$.
\end{proof}

In the {\em (weighted) Paging with Time Windows and Penalties}
(\wPwTwP) problem, each request interval $I$ has an associated penalty
value $\ell(I) \geq 0$, which is the penalty for not satisfying the
request associated with interval $I$. Not all requests must be
satisfied, but if some request is not satisfied we must pay the
penalty for it. The IP for the \wPwTwP problem is very similar, where
$y_I$ is the indicator for whether we choose to take the penalty.
\begin{alignat}{2}
  \min \sum_{p,t} w(p) \,  x_{p,t} + \sum_I \ell(I) \, & y_I \tag{IPp} \label{eq:IPp}\\
  \sum_{I \in \calC} \Big( y_I + \sum_{t' \in \rext{I,t}}
  x_{\page(I),t'} \Big) &\geq 1
  & \quad &\forall t \, \forall \calC \text{ with $k+1$ requests for
    distinct pages starting
    before $t$}  \tag{R1} \label{eq:1p} \\
  \sum_{I \in \calC} \Big(y_I + \sum_{t' \in \dext{I,t}} x_{\page(I),t'}\Big)
  &\geq 1 - y_{I_t}
  & \quad &\forall t \, \forall \calC \text{ with $k$ requests for distinct
    pages  (excluding $p_t$)  ending before $t$}  \tag{D1} \label{eq:2p} 
\end{alignat}
This IP is a strict generalization of~(\ref{eq:IP}), since we can set
$\ell(I) = \infty$ to force the $y_I = 0$. The proof of validity
of~(\ref{eq:IPp}) for the \wPwTwP problem is identical to
\Cref{cl:valid} above, and is omitted.

\section{Solving \wPwTw and \wPwTwP Online using Online Solutions to (\ref{eq:IP}) and (\ref{eq:IPp})}
\label{sec:solving-wpwtw-using}
\newcommand{\Sat}{\mathsf{Sat}}

Now that we have the IPs, we need to solve them online, and also show how to convert a solution 
into one for the \wPwTw or \wPwTwP problem. Indeed, (\ref{eq:IP}) and (\ref{eq:IPp}) do 
not have any explicit capacity constraints, so we need to extract a ``schedule'' from
the IP solution in an online manner. In this section we show the latter step; we discuss 
solving the IPs in \Cref{sec:first-lp}.

\begin{theorem}
  \label{thm:round-on}
  There is an online algorithm that converts an
  $\alpha$-competitive integral solution to~(\ref{eq:IPp}) into a
  valid $O(\alpha\log n)$-competitive solution for the \wPwTwP instance.
  As a special case, there is an online algorithm that converts an
  $\alpha$-competitive integral solution to~(\ref{eq:IP}) into a
  valid $O(\alpha\log n)$-competitive solution for the \wPwTw instance.
\end{theorem}

Although we stated this theorem in terms of the more general \wPwTwP
problem, we will actually prove it for the simpler \wPwTw problem.
This is without loss of generality since an integer solution to 
(\ref{eq:IPp}) for an instance of the \wPwTwP problem 
already specifies the page requests that are being 
satisfied by the solution, and the ones where the solution incurs
the penalty. Given an instance of the \wPwTwP problem, 
we simply remove the requests in the latter set to 
create an equivalent instance of the \wPwTw problem. On this 
instance, we
apply the above theorem for the \wPwTw problem to recover the 
theorem for the \wPwTwP instance.

In the rest of the paper, we move between solutions $x$
to~(\ref{eq:IP}) and their characteristic set
$A^\star := \{(p,t) \mid x_{p,t} = 1\}$.  Visually, thinking of time
as the $x$-axis and the $n$ pages as the $y$-axis, the solution
$A^\star$ corresponds to a set of ``stars'' in the $2$-dimensional
plane.  Let us list some properties of the online solution $\astar_t$
to~\eqref{eq:IP} (which should satisfy all the constraints
corresponding to times $t$ and earlier) that are maintained by the
algorithm in \Cref{sec:first-lp}.  
\begin{OneLiners}
\item[(A1)] \emph{Monotonicity}: $\astar_t \subseteq \astar_{t+1}$ for
  all $t$.
\item[(A2)] {\em Past-Preservation}: At time $t$ the algorithm only
  adds stars corresponding to times $t$ or later. Ideally, at time $t$,
  it should only add stars at time $t$, with the following exception.
\item[(A3)] {\em Sparsity}: for every page $p$, $\astar_t$
  contains at most one star $(p,t')$ with $t' > t$. Furthermore, if
  $\astar_t$ has such a star, then this star  hits all the request intervals
  for $p$ which contain time $t$. In fact, our online algorithm for
  \wPwTw does not need to know the exact location of the stars after
  time $t$---it just needs to know the set of pages $p$ for which the
  solution $\astar_t$ contains such a star.
\end{OneLiners}

\medskip
The main idea of the algorithm is that if the cache is full and we need
to evict a heavy page $p$, we should spend about $w(p)$ amount of weight
in serving other outstanding requests at time $t$. The requests that
need to be serviced need to be carefully chosen, because there are
conflicting goals: (i) we want to service the cheaper requests,
because this way we can service many of these, (ii) we want to go by EDF
(Earliest Deadline First) order because the ones ending soon are more
critical, and finally (iii) we prefer to service the requests which are
hit by the solution $A^\star_t$ because we can directly pay for these service
costs. Interestingly, we show that we can simultaneously take care of
all of these three requirements. Moreover, we can identify a weight $w$
such that we can take care of {\em all} outstanding requests which are
cheaper than $w$ and are not hit by $A^\star_t$.

\subsection{The Online Algorithm for Non-Overlapping Requests}
\label{sec:non-overlap}

\begin{figure}[t]
  \begin{algorithm}[H]
    \ForEach{$t=0, 1, \ldots$}{
      \textbf{let} $I_t$ be the interval with deadline $t$, and let $p_t
      \gets \page(I_t)$ \\
      \label{l:firsttest} \If{cache $C(t)$ full and $I_t$ not satisfied}{
        \textbf{evict} the least-weight page $p_{\min}$ in
        $C(t)$ \label{l:evict1} \\
        \If{$w(p_t) \leq 2\,w(p_{\min})$}{

          \medskip
          $Z^\star \gets \emptyset$. \label{l:Zb1} \\
          \For{every page $p$ in $C(t)$}{
            $I^p_t \gets$ the request interval $I$ with $\page(I) = p$ and
            largest ending time $e(I) < t$. \label{l:Ipee} \\
            \textbf{if} $\dext{I^p_t,t}$ is hit by $\astar_t$ then add %
            $p$ to
            $Z^\star$. \label{l:Ze1}
          }
          
          $U \gets$ unsatisfied request intervals active at time $t$ %
          (one per page, page requests are disjoint%
          ). \label{l:Ud}
          \\
          $U^\circ \gets \{I \in U \mid \nexists t' \in I \text{ with }
          (\page(I),t') \in \astar_t \}$ be intervals in $U$ \textbf{not} hit by
          $\astar_t$ \label{l:Ucirc}\\

          \medskip
          \textbf{serve} and \textbf{evict} all requests in $U \setminus
          U^\circ$. \label{l:serve1} \\ 

            \medskip

            \textbf{let} $U^\circ_{\leq w}$ and $Z_{\leq w}^*$ denote pages in $U^\circ$ and $Z^\star$
            respectively with weight at most $w$. \\
            \textbf{let} $\pstar$ be a page in $Z^\star$ such that $w(U^\circ_{\leq
              2w(\pstar)}) \leq 2 \cdot w(Z^\star_{\leq w(\pstar)})$. \label{l:goodp}
            \\ 
            \textbf{evict} all pages in $Z^\star_{\leq
              w(\pstar)}$. \label{l:evict2} \\
            \textbf{serve} and \textbf{evict} all requests in $U^\circ_{\leq 2w(\pstar)}$. \label{l:serve2} \\
        }
      }
      \textbf{if} $I_t$ not satisfied \textbf{then}
      bring page $p_t$ into
      cache. 
      \label{l:retain}
    }
    \caption{ConvertOnline$(\text{Online~\eqref{eq:IP} solution }\astar_t)$}\label{algo:first}
  \end{algorithm}
  \caption{Online Algorithm to Service Request Intervals}
  \label{fig:algo}
\end{figure}
  
In this section, we assume that no two request intervals for the same
page overlap---that is, for any pair of requests $I, I'$ for the same page,
$I \cap I' = \emptyset$. This gives a simpler algorithm than for
the general case, which follows the same approach but has to deal with
the case that multiple request intervals for the same page may try to
charge to the same star in $A^\star_t$. (See
\S\ref{sec:online-conversion} for the online algorithm for the general
case, and
\Cref{sec:overlap}
for the offline algorithm).

\Cref{algo:first} (see \Cref{fig:algo})
shows how to convert an online solution $\astar_t$ to~\eqref{eq:IP} 
into a feasible solution to the underlying \wPwTw instance. At each
time~$t$, we begin with some pages $C(t)$ in the cache. If the unique
request $I_t$ ending at time $t$ is not already satisfied, and the cache
is full, we evict the cheapest page $p_{\min}$ in the cache. We then potentially serve some other pending requests by
bringing in and then evicting them, and also potentially remove some
other pages from the current cache. (These services and removals help
pay for evicting $p_{\min}$.) 
\begin{figure}
    \centering
    \includegraphics[width=5in]{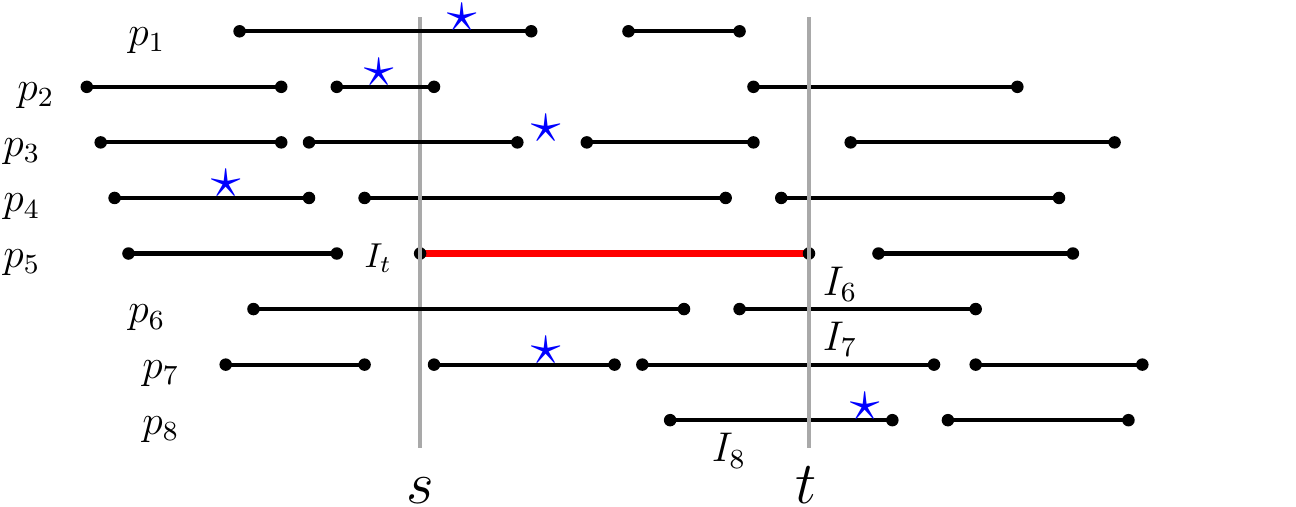}
    \caption{\small\emph{Illustration of the definitions used in
        Algorithm~\ref{algo:first}. The request intervals for each page
        are shown in one horizontal line (these do not overlap in our
        example). Focus on time $t$: the cache has pages
        $C(t)=\{p_1, \ldots, p_4\}$. The page $p_5 = \page(I_5)$ is
        critical at time $t$, with critical interval $I_t=[s,t]$. The
        solution $A^\star$ is given by the stars, each of which
        corresponds to a star $(p,t)$ in the natural manner. The set
        $Z^\star = \{p_1, p_2, p_3\}$ and assuming $I_6, I_7, I_8$
        are unsatisfied at time $t$, the set $U^\circ = \{I_6,
        I_7\}$.}}
    \label{fig:algoex}
\end{figure}

Specifically, for every page $p$ in $C(t)$, define $I^p_t$ to be the
most recent request for $p$ which ends before $t$ --- this is
well-defined because requests don't overlap.
Define $Z^\star$
to be the pages in $C(t)$ for which the interval $\dext{I^p_t,t}$ is hit
by $\astar_t$. (Since $\dext{I^p_t, t}$ ends at time $t$, this only
requires the knowledge of stars in $\astar_t$ at or before time $t$). We
can directly pay for evicting these pages from the
cache. But the situation is tricky---some $\dext{I^p_{t'},t'}$ for a future
time $t'$ may also be hit by the same star in 
$A^\star_t$. So we evict a subset of $Z^\star$---ones for which we
are sure that the corresponding stars of $A^\star_t$ won't be charged
again in the future.

To do this, the first simple observation is that we need to do this charging
only when the critical page is not much heavier than the cheapest page
in the cache, else we can charge the eviction to the much heavier page in the
cache. We define $U$ to be the set of outstanding requests at
time $t$ and $U^\circ$ to be the subset of $U$ which are not hit by
$\astar_t$ (lines~\ref{l:Ud}--\ref{l:Ucirc})---by the sparsity property,
these are the pages $p$ for which $\astar_t$ does not currently have a star beyond time $t$. We service all the requests
in $U \setminus U^\circ$ immediately (we service  a request by loading the corresponding page in the cache, and ``evict a request'' by evicting the corresponding page from the cache)---these request intervals are hit by $A^\star_t$ and
can be directly paid for (because of the non-overlapping intervals). It is
trickier to decide which requests in $U^\circ$ to service. In
\Cref{lem:Z} 
we show there is a page $\pstar$ in $Z^\star$ such that
$w(U^\circ_{\leq 2w(\pstar)}) \leq 2 \cdot w(Z^\star_{\leq w(\pstar)})$,
where the notation $X_{\leq a}$ denotes all the stars in $X$ of
weight at most $a$. We service all the requests in
$U^\circ_{\leq 2w(\pstar)}$ and evict the pages
$Z^\star_{\leq w(\pstar)}$. This ensures that all the remaining
unsatisfied requests are much heavier than the current pages remaining
in the cache. By the observation at the start of this paragraph, the
stars in $A^\star_t$ which are being charged for the eviction of
$Z^\star_{\leq w(\pstar)}$ are not going to be charged again.

Finally, we serve $I_t$ by bringing $p_t = \page(I_t)$ into the cache if
still needed. Observe that the cache $C(t+1)$ at the start of time $t+1$
is contained within $C(t) \cup \{p_t\}$, since all other pages we
satisfy at time $t$ are also evicted. Moreover, if $C(t)$ was full, the cheapest page in $C(t)$ is evicted, and other pages from
$C(t)$ may be evicted too.

\subsubsection{The Analysis}

We first need some
supporting claims to show that the algorithm is well-defined, and then
bound the cost. 

\begin{restatable}{claim}{OnlineBetween}
  \label{cl:0}
  Suppose a page $p$ is evicted from the cache at time $t_1$ but is in
  the cache at the end of time $t_2 > t_1$. Then there must exist a
  request interval $I$ for page $p$ with $t_1 < s(I) \leq e(I) \leq t_2$.
\end{restatable}

\begin{proof}
  The only step in the algorithm when a page is brought into the cache
  (and not evicted immediately afterwards) is
  line~\ref{l:retain}. Hence, there must be an unsatisfied request interval $I$
  ending at some time $e(I) \leq t_2$ which brought in (and kept) $p$ in
  the cache after it had been evicted at time $t_1$. If $s(I) \leq t_1$, then $I$ would have been satisfied at time $t_1$, which is not true. 
\end{proof}
We now show that the algorithm is well-defined. 
\begin{claim}
  \label{cl:Z}
  The set $Z^\star$ defined in lines~\ref{l:Zb1}--\ref{l:Ze1} is non-empty.
\end{claim}
\begin{proof}
  For each page $p \in C(t)$, let $I^p_t$ be the request interval defined
  in line~\ref{l:Ipee}---such a request interval exists because of
  Claim~\ref{cl:0} (we assume that the cache is empty initially). Applying the IP constraint~(\ref{eq:2}) to time $t$ and
  these $k$ request intervals implies that at least one of their
  doubly-extended intervals is hit by $\astar_t$, and hence the corresponding page belongs to $Z^\star$.
\end{proof}

\begin{lemma}
  \label{lem:Z}
  There exists a page %
  $\pstar \in Z^\star$ such that 
  \[ w(U^\circ_{\leq 2w(\pstar)}) \leq 2\, w(Z^\star_{\leq
      2w(\pstar)}). \]
\end{lemma}

\begin{proof}
  We first claim that $|U^\circ| \leq |Z^\star|$. Indeed, define a set
  of $k+1$ request intervals as follows. For each page
  $p \in C(t) \setminus Z^\star$, consider the request interval $I^p_t$
  for page $p$ as defined in line~\ref{l:Ipee}. Since %
  $p \not\in Z^\star$, we must have $(p,t')\not\in
  \astar_t$ for all times $t' \in
  \dext{I^p_t,t}$. But since the interval ends before
  $t$, we get $\rext{I^p_t, t} \sse \dext{I^p_t,t}$ and so
  $\astar_t$ does not hit the right-extended interval for
  $I^p_t$ either. To this collection of
  $k-|Z^\star|$ intervals, add the request intervals corresponding to
  $U^\circ$---all these request intervals contain
  $t$, so the right-extension operation does not extend them. Moreover,
  we have at most one interval per page, and they are all unsatisfied,
  so the collection now has $|U^\circ| +
  k-|Z^\star|$ many intervals for distinct pages. And none of their
  right-extensions are hit by
  $\astar_t$, so by constraint~(\ref{eq:1}) this collection has size at
  most $k$. This proves that $|U^\circ| \leq |Z^\star|$.

  Let $A$ be the set of pages in $Z^\star$ and $B$ the set of pages in
  $U^\circ$. We set up a bipartite graph on $(A,B)$ with an
  edge between $p \in A$ and $p' \in B$ if $w(p') \leq 2\, w_p$. If this
  graph has a perfect matching, then $w(B) \leq 2\, w(A)$. We choose $\pstar$
  to be the highest-weight page in $Z^\star$.

  Else such a perfect matching does not exist. Let $A' \subseteq A$ be a
  minimal Hall set, and $B'$ be the neighborhood of $A'$. Let $a$ be any
  page in $A'$. The pages in $A' \setminus \{a\}$ can be matched with
  $B'$. Therefore, $w(B') \leq 2\,w(A' \setminus \{a\}) \leq 2w(A')$. Now
  choose $\pstar$ to be the highest weight page in $A'$, to get
  $w(U^\circ_{\leq 2w(\pstar)}) \leq 2\, w(Z^\star_{\leq
    w(\pstar)})$.
\end{proof}

Therefore when we
reach line~\ref{l:goodp}, a page $\pstar$ of the desired form exists, and
the algorithm is well-defined. 
Finally the next claim shows that the request interval $I_t$ gets 
served.

\begin{restatable}{claim}{OnlineConvertFeas}
  \label{clm:pt-evicted}
  Suppose $I_t$ is unsatisfied at time $t$. If $w(p_t) \leq
  2w(p_{\min})$, then the page $p_t$ belongs to either $U\setminus
  U^\circ$ in line~\ref{l:serve1} or to $U^\circ_{\leq 2w(\pstar)}$ in
  line~\ref{l:serve2}, and is served and evicted. Else $p_t$ is served
  by line~\ref{l:retain}, and remains in the cache.
\end{restatable}

\begin{proof}
  The page $p_t$ is unsatisfied at the beginning of time $t$ and so $I_t$
  belongs to $U$. First assume $w(p_t) \leq 2w(p_{\min})$. If $I_t$ is not in
  $U^\circ$, it is served in line~\ref{l:serve1}. Else we reach
  line~\ref{l:serve2}. In this case, the page $\pstar$ chosen in
  line~\ref{l:goodp} has $2w(\pstar) \geq 2w(p_{\min}) \geq w(p_t)$, so $p_t$
  beings to $U^\circ_{\leq 2w(\pstar)}$ and is served in line~\ref{l:serve1}.
  The other case (where $w(p_t) > 2w(p_{\min})$) is clear by construction.
\end{proof}

\subsubsection{The Cost Guarantee}  
\label{sec:costg}
We want to bound the total cost incurred till time $T$. 
The high-level cost analysis goes as follows. If the cache has room we
can just satisfy $I_t$, so suppose the cache is full and we need to pay
to evict $p_{\min}$. If the page $p_t$ %
is twice as heavy as
$p_{\min}$, we can charge $p_{\min}$ to $p_t$ and pay when $p_t$ is
subsequently evicted. Else, if the unsatisfied intervals crossing time
$t$ which are hit by $\astar_t$ have large weight, i.e., if $w(U\setminus
U^\circ) \geq w(p_{\min})$,
we can serve and evict them and then charge to them---this can pay for
$p_{\min}$. Finally, we evict some pages from the current cache that are
hit by $\astar_t$: they pay for both evicting $p_{\min}$ and for serving
some more of the outstanding requests. These pages are evicted from the
cache to ensure they are not charged again. 

We now show how to pay for the possible evictions in
lines~\ref{l:evict1}, \ref{l:serve1}, and~\ref{l:evict2}-\ref{l:serve2}
(since bringing in pages is for free). We maintain the invariant that
each page $p$ in the cache has at most $w(p)$ ``load'' on it; pages
outside the cache have zero load. The load measures the evictions which have not been paid for till now.  If we bring in $p_t$ and if $w(p_t)
\geq 2w(p_{\min})$, its load becomes the load of $p_{\min}$ plus the
cost of evicting $p_{\min}$; thus the total load on $p_t$ is at most its
weight, and $p_t$ remains in the cache, maintaining the invariant. At
the end of the algorithm, the total load over all pages is at most the
weight of the pages in the cache, which is at most the optimum
cost. This adds one to the competitive ratio. Else if $w(p_t) <
2w(p_{\min})$, we evict at least one page in $Z^\star$ (by
\Cref{cl:Z})
and can charge evicting $p_{\min}$ to the eviction of that page---which
we show below how to charge to $\astar_T$.

Next: we charge evicting $U\setminus U^\circ$ to $w(\astar_T)$ as
follows. Each interval in $U\setminus U^\circ$, say for page $p$,
contains some star at $(p,t')$ in $\astar_t$ (and hence in $\astar_T$, and we can charge to star). Since
the request intervals for a page are disjoint (by our simplifying
assumption),  $(p,t)$ cannot lie in any other request interval for
$p$, and will not be charged by line~\ref{l:serve1} again.

Finally, we charge  the eviction cost for lines
\ref{l:evict2}-\ref{l:serve2}. This cost  is $O(w(Z^\star_{\leq w(\pstar)}))$ by
our choice of $\pstar$ in line~\ref{l:goodp}. Observe that for each page $p$
$I$ in $Z^\star_{\leq w(\pstar)}$, the doubly-extended interval $\dext{I^p_t,t}$
is hit by an star at $(\page(I),t') \in \astar_t$ for some $t' \leq t$, so we want to charge to this
star of $\astar_t$. Moreover, each page in $Z^\star_{\leq w(\pstar)}$ is at
least as heavy as $p_{\min}$, so any of these stars of $\astar_t$ can
pay to evict $p_{\min}$ (and its load).  We finally show that no star
of $\astar_T$ can be charged twice in this manner.
\begin{lemma}
  \label{lem:twice}
  No star in $\astar_T$ can be charged twice because of evictions in
  lines~\ref{l:evict2}-\ref{l:serve2}.
\end{lemma}
\begin{proof}
  For a contradiction, suppose an star at $(q,t_{q}) \in \astar_T$ is charged
  twice, at time $t_1$ and time $t_2$. Hence, at both these times $q$
  was in the cache and was evicted in line~\ref{l:evict2}, so all
  unsatisfied pages that were active at these times and had weight
  $\leq 2w(q)$ were definitely served by line~\ref{l:serve2}. (We will
  contradict this implication of our assumption.)

  Let $I^{q}_{t_1}$ and $I^{q}_{t_2}$ be the corresponding intervals
  defined in line~\ref{l:Ipee} for the page $q$.  
  \Cref{cl:0} 
  shows
  that $I^{q}_{t_2}$ starts after $t_1$. But we know that
  $t_{q} \leq t_1$, since $(q,t_q)$ was charged at time $t_1$, and so
  $t_q \not \in I^q_{t_2}$. So, in order for $(q,t_q)$ to hit the
  doubly-extended interval $\dext{I^q_{t_2}, t_2}$, it must be the case that 
  the critical interval $I_{t_2}$  contained the time $t_{q}$ (and
  hence time $t_1 \in [t_q, t_2]$). Let $p_{t_2}$ denote $\page(I_{t_2})$. 
  Then $w(p_{t_2}) \leq 2w(q)$,
  else we would merely have evicted the cheapest page at time $t_2$ and
  not reached lines~\ref{l:evict2}-\ref{l:serve2} again. This means
  $p_{t_2}$ had weight at most $2w(q)$, and the request $I_{t_2}$  was active and remained
  unsatisfied at the end of time $t_1$, which contradicts the
  implication above.
\end{proof}

This proves 
\Cref{thm:round-on} 
(without losing the extra $\log n$
factor) in the case of non-overlapping requests for any page $p$.  The
general case %
gets trickier. Indeed, consider the example with a page $p$ having
request intervals $[t_1, t], [t_2, t], \ldots, [t_k,t]$, where $t_1 <
t_2 < \ldots < t_k < t. $ Suppose we have a star $(p,t) \in
A^\star_{t_1}$. Consider a time $t' \in [t_1, t]$ when the algorithm
reaches line~\ref{l:Ud}. If any of these intervals is not satisfied at
$t'$, then they will get counted in $U \setminus U^\circ,$ and so we
will charge the star at $(p,t)$ for servicing $p$ at time $t'$. But this
can happen for multiple values of $t'$, and we have only one star in
$A^\star$ to charge to. Moreover, we cannot say that we will take care
of all these requests at the ending time $t$---since all the pages in
the cache may be very expensive at that time. In the off-line case
(which appears in 
\Cref{sec:overlap}),
one can add a {\em reverse
  delete} step, where we look at all these times when we service some of
these requests, and realize that a subset of them would
suffice. However, we discuss the more involved online case in the next section.

\subsection{Online Algorithm for the General Setting}
\label{sec:online-conversion}

The algorithm from 
\Cref{sec:non-overlap} 
assumes the requests for a
page are non-overlapping. We now
extend it to handle overlapping requests in an online fashion.
\Cref{algo:first-online} (see \Cref{fig:algo-online}) gives the online algorithm---the lines changed  from \Cref{algo:first} are highlighted. 
We call a
request interval $I$ \emph{non-dominating} if it does not contain
another request interval for $page(I)$---we know whether $I$ is
non-dominating only at time $e(I)$. Notice that the definition of
$I^p_t$ in line~\ref{l:Ipee-online} looks only at non-dominating intervals.  

Since the
request intervals for a particular page are no longer disjoint, we do not serve and evict all
the intervals in $U\setminus U^\circ$ when we create space at time $t$
(as \Cref{algo:first} would do in \cref{l:serve1}).
 Instead we only serve the requests hit by $\astar_t$
before time $t$, and some small set of requests that are hit by $\astar_t$
after time $t$. As shown in the example at the end of
\Cref{sec:non-overlap}, 
serving all such requests may lead to unbounded number of chargings to a star in $\astar_T$. These requests are considered in the earliest deadline order and their total weight is a constant times the weight of the cheapest page in $Z^\star$ (denoted by $\pdag$). It is also worth noting that we perform these steps only if $I_t$ is hit by $\astar_t$ (line~\ref{l:ustar-contains-pt}).

It is also worth noting that line~\ref{l:sort} is the only place in the algorithm where we need to know the right end-point of an existing request for a page.

\begin{figure}[ht]
  \removelatexerror
  \begin{algorithm}[H]
    \ForEach{$t=0, 1, \ldots$}{
      \textbf{let} $I_t$ be the interval with deadline $t$, and let $p_t
      \gets \page(I_t)$ \\
      \label{l:firsttest-online} \If{cache $C(t)$ full and $I_t$ not satisfied}{
        \textbf{evict} the least-weight page $p_{\min}$ in
        $C(t)$ \label{l:evict1-online} \\
        \If{$w(p_t) \leq 2\,w(p_{\min})$}{

          \medskip
          $Z^\star \gets \emptyset$. \label{l:Zb1-online}\\
          \For{every page $p$ in $C(t)$}{
            $I^p_t \gets$ \hl{non-dominating} request interval $I$ with $\page(I) = p$ and
            largest ending time $e(I) < t$. \label{l:Ipee-online} \\
            \textbf{if} $\dext{I^p_t,t}$ is hit by $\astar_t$ then add $I^p_t$ to
            $Z^\star$. \label{l:Ze1-online}
          } 
          $U \gets$ unsatisfied request intervals active at time $t$
          (one per page, \hl{if there are multiple choose one with earliest deadline}). \\
          $U^\circ \gets \{I \in U \mid \nexists t' \in I \text{ with }
          (\page(I),t') \in \astar_t \}$ be intervals in $U$ \textbf{not} hit by
          $\astar_t$.\\

          \medskip
            \If{\hl{$I_t \not\in U^\circ$}}{ \label{l:ustar-contains-pt}
              \hl{$U^\star_t \gets$ request intervals $I$ in $(U
                \setminus U^\circ)$ which are
              hit by $\astar_t$ at some time $\leq t$}.
             \label{l:ustar} \\
              \hl{\textbf{serve} and \textbf{evict} all requests in
                $U^\star_t$}. \label{l:serve1-online} \\
              \hl{\textbf{evict} the cheapest page $\pdag$ in
                $Z^\star$ (this may be the same as $p_{\min}$)} \label{l:evict2-online} \\
              \hl{\textbf{sort} intervals in $U \setminus (U^\circ
                \cup U_t^\star)$ with
              weights  $\leq 2w_{p^\dagger}$ in ascending order of end-times}.
              \label{l:sort} \\
              \hl{\textbf{serve} and \textbf{evict} a maximal prefix of these intervals
                with total weight at most
                $4w_{p^\dagger}$}. \label{l:serve2-online} \label{l:case1e}
            }
            \textbf{let} $U^\circ_{\leq w}$ and $Z_{\leq
              w}^\star$ denote pages in $U^\circ$ and $Z^\star$
             with weight at most $w$.   \\
            \textbf{let} $\pstar$ be a page in $Z^\star$ such that
            $w(U^\circ_{\leq
              2w(\pstar)}) \leq 2 \cdot w(Z^\star_{\leq w(\pstar)})$. \label{l:goodp-online}
            \\ 
            \textbf{evict} all pages in $Z^\star_{\leq
              w(\pstar)}$. \label{l:evict3-online} \\
            \textbf{serve} and \textbf{evict} all requests in $U^\circ_{\leq 2w(\pstar)}$. \label{l:serve3-online}
      }
      }
      \textbf{if} $I_t$ not satisfied \textbf{then}
      bring page $p_t$ into
      cache.
      \label{l:retain-online}
    }
    \caption{ConvertOnline(\eqref{eq:IP} solution $\astar_t$ appearing online)}\label{algo:first-online}
  \end{algorithm}
  \caption{Online Algorithm to Service Request Intervals}
  \Description{Online Algorithm}
  \label{fig:algo-online}
\end{figure}

By 
\Cref{lem:Z} 
the page $\pstar$ in 
\cref{l:goodp-online}
exists,
and hence the algorithm is well-defined. For the correctness we need to
show that each page is served. Indeed, for an unsatisfied request $I_t$
at its deadline $t$, either $p_t$ has twice the weight of $p_{\min}$ and is
handled in 
\cref{l:retain-online}.
Else, either the request interval
$I_t$ is hit by $\astar_t$ and so it belongs to $U_t^\star$ and is served/evicted in
\cref{l:serve1-online}, 
or it is not hit by $\astar_t$ and so
belongs to $U^\circ_{\leq 2w(\pstar)}$ (by
\Cref{clm:pt-evicted}) 
and is
served/evicted in 
\cref{l:serve3-online}. 
It remains to estimate the total eviction cost of this algorithm. 

\subsubsection{The Cost Analysis}
\label{sec:onlinecostanalysis}

Consider the run of the algorithm until time $T$; we bound the total
cost incurred until this time, in terms of $w(\astar_T)$. Observe that
evictions can only happen on lines~\ref{l:evict1-online},
\ref{l:serve1-online}--\ref{l:serve2-online}, and
\ref{l:evict3-online}--\ref{l:serve3-online}.  The first and the last of
these can be dealt with as in 
\Cref{sec:non-overlap}. 
Indeed, paying for
$p_{\min}$ and its load is done by either putting a load on $p_t$ if
$w(p_t) \geq 2 w(p_{\min})$ or else at least one other page from $C(t)$
is evicted and charged for, and we can handle $p_{\min}$ by charging a
constant factor more. The total cost incurred during
lines~\ref{l:evict3-online}--\ref{l:serve3-online} is at most
$2 w(\astar_T)$, since the proof for 
\Cref{lem:twice} 
remains unchanged.
Indeed, at each time $t$ when we perform those evictions, we charge the stars in $A^\star_T$ which hit the
intervals in $Z^\star_{\leq w(\pstar)}$, 
and these stars are never charged again due to
\Cref{lem:twice}. 

It remains to bound the cost incurred during
lines~\ref{l:serve1-online}--\ref{l:serve2-online}.  Let $\fT \sse [T]$ contain
the times when we reach those lines. Let
$Z^\star_t$, $U_t$, $U^\circ_t$, and $U^\star_t$ denote the corresponding sets at
time $t$, and $\pdag_t, \pstar_t$ denote the pages chosen in
\cref{l:evict2-online} and~\cref{l:goodp-online}. 
The
evictions in 
\cref{l:serve1-online} are easy to pay for:
\begin{restatable}{claim}{OnlineEasy}
  \label{cl:u1}
  $\cup_{t \in \fT}\; w(U^\star_t) \leq w(\astar_T)$.
\end{restatable}
\begin{proof}
  We charge the weight of evicting $p=\page(I)$ for some request interval $I$ in $U^\star_t$
  to the element $(p,t') \in \astar_t$ where $t' \in [0, t] \cap I$.  We
  claim that no element in $\astar_T$ will get charged twice this
  way. Indeed, if we charge to $(p,t')$ at time $t \geq t'$, we have
  satisfied all existing request intervals for page $p$ containing the time
  $t'$---and no future requests can arrive that contain it.
\end{proof}

Some more notation: let
$\Udag_t$ be the prefix of request intervals serviced in
\cref{l:serve2-online}.
For a time $t \in \fT$, define the {\em effective cost} at time $t$ to
be $w(\pdag_t) + w(\Udag_t)$---this is the remaining cost incurred at time $t$ in
\cref{l:evict2-online,l:serve2-online}.   
For an interval $[a,b]$, let $\astar_t[a,b]$
denote the set of $(p,t')$ in $\astar_t$ where $t' \in [a,b]$.
Let $P^\star_t[a,b]$ be the set of pages corresponding to which there is at
least one star in $\astar_t[a,b]$. Note that $w(P^\star_t[a,b]) \leq w(\astar_t[a,b])$ for any time $t$ and interval $[a,b]$. 

\begin{claim}
  \label{cl:case1alt}
  Suppose times $t_1, t_2 \in \fT$ are such that $t_1 < t_2$ and
  $I_{t_2}$ contains time $t_1$. Then the effective cost at time
  $t_1$ is at most at most $\frac52 \, w(P^\star_{t_2}[t_1,t_2])$. 
\end{claim}
\begin{proof}
  By design, $w(\Udag_{t_1}) \leq 4 w(\pdag_{t_1})$, so the effective
  cost at time $t_1$ is at most $5w(\pdag_{t_1})$. Thus it suffices to
  bound  $w(\pdag_{t_1})$.
  Since $t_2 \in \fT$, the interval $I_{t_2}$ was not served at time
  $t_1$. The possible reasons are:
  \begin{enumerate}
  \item The interval $I_{t_2} \in U^\circ_{t_1}$ and $w(p_{t_2}) >
    2\,w(p^\star_{t_1})$. Since $w(\pdag_{t_1}) \leq
    w(\pstar_{t_1})$ by the choice of $\pdag_{t_1}$, so $w(\pdag_{t_1})
    < \frac12 w(p_{t_2})$. Since $I_{t_2}$ is hit by $\astar_{t_2}$
    (because $t_2 \in \fT$), the past preserving property of the online
    solution $\astar$ implies that there must be a star at $(p_{t_2}, t')$ in $\astar_{t_2}$ for some $t' \in [t_1, t_2]$. Therefore, $w(p_{t_2}) \leq 
    w(P^\star_{t_2}[t_1, t_2]).$
  \item $I_{t_2} \not\in U^\circ_{t_1}$ but
    $w(p_{t_2}) > 2 w(\pdag_{t_1})$, so it was not considered in the sorted
    ordering (in line~\ref{l:sort}). However, $I_{t_2}$ was not in $U^\star_{t_1}$, so it was
    hit by $\astar_{t_1}$ at some time after $t_1$---this means
    $w(p_{t_2})$ is counted in $w(P^\star_{t_1}[t_1,t_2])$. So
    $w(\pdag_{t_1}) \leq \frac{1}{2} w(P^\star_{t_1}[t_1,t_2])$.
  \item $I_{t_2} \not\in U^\circ_{t_1}$ and $w(p_{t_2}) \leq 2
    w(\pdag_{t_1})$ but we did not 
    add $I_{t_2}$ to $\Udag_{t_1}$ at time $t_1$: So  $w(\Udag_{t_1})$ must
    have been more than
    $4\, w(\pdag_{t_1}) - w(p_{t_2}) \geq 2\, w(\pdag_{t_1})$. We added
    intervals to $\Udag_{t_1}$ in the earliest-deadline-first order, so
    all the weight added to $w(\Udag_{t_1})$ before considering
    $I_{t_2}$ belongs to $w(\astar_{t_1}[t_1, t_2])$. Chaining these inequalities,
     $w(\pdag_{t_1}) \leq \frac12 w(\Udag_{t_1}) \leq \frac{1}{2} w(P^\star_{t_1}[t_1,t_2])$.
   \end{enumerate}
  Since $\astar_{t_1} \sse \astar_{t_2}$, we get the desired result. 
\end{proof}

\begin{claim}
  \label{cl:case1alt'}
  Suppose $t_1, t_2 \in \fT$ are such that $t_1 < t_2$, and $I_{t_2}$
  \emph{does not} contain $t_1$. Suppose
  $\pdag_{t_1} = \pdag_{t_2}$---call this page $\pdag$---then there is
  an star for $(\pdag,t')$ in $\astar_{t_2}$ for some
  $t' \in (t_1, t_2]$.
\end{claim}

\begin{proof}
  Page $\pdag$ is evicted at time $t_1$, and so it must have
  been brought in by an unsatisfied request $I$; this request must start
  after $t_1$ (else it would be satisfied at $t_1$) and end before $t_2$
  (since $\pdag$ is in the cache at time $t_2$). We claim that $I^{\pdag}_{t_2}$ 
  is also contained in $(t_1, t_2]$. Suppose not. So $s(I^{\pdag}_{t_2}) \leq t_1.$ 
  
  The interval $I$ is either itself non-dominating, or contains a non-dominating request for $\pdag$. In either case, there is a non-dominating request interval for $\pdag$ which is contained in $(t_1,t_2]$---call this $I'$ (it could be same as $I$). Now, $I'$ is not designated as $I^{\pdag}_{t_2}$. It must be the case that 
  $t(I^{\pdag}_{t_2}) \geq t(I')$. But then $I^{\pdag}_{t_2}$ contains $I'$, which contradicts the fact that it is non-dominating. 
  
  Since $I_{t_2}$ is also contained in $(t_1, t_2]$, we see that 
  $\dext{I^{\pdag}_{t_2}, t_2}$ is also contained in $(t_1, t_2]$. Since
  $\pdag \in Z^\star$ at time $t_2$, $\dext{I^{\pdag}_{t_2}, t_2}$
   is hit by $A^\star_{t_2}.$ This proves the claim. 
\end{proof}

\paragraph{The Charging Forest.}
Motivated by Claims~\ref{cl:case1alt} and~\ref{cl:case1alt'}, we define
a directed forest $F=(\fT,E)$ as follows. For time $t \in \fT$, if time
$t'$ is the smallest time such that $t' > t$ and the critical interval
$I_{t'}$ for $t'$ contains $t$, we define $t'$ to be the parent of $t$,
i.e., we add an arc $(t,t')$. If no such time $t' > t$ exists, then $t$
has no parent (i.e., zero out-degree). The following lemma gives some
natural properties of the forest 
$F$.

\begin{figure}
    \centering
    \includegraphics[width=4in]{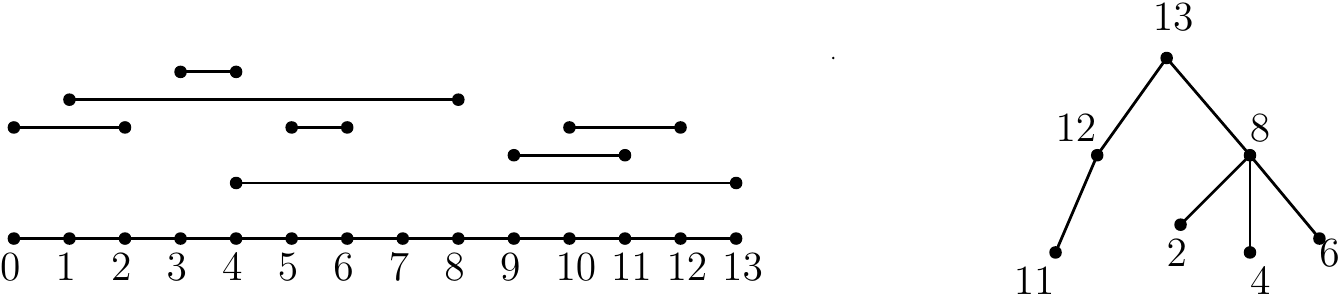}
    \caption{\small\emph{Illustration of $F$: the intervals on the left are $I_t$ for $t \in \fT$ (note that these intervals are identified using their right end-points). The corresponding forest $F$ is shown on the right.}}
    \label{fig:forest}
\end{figure}

\begin{restatable}{lemma}{LemCharging}
  \label{lem:forest} Suppose $t,t' \in \fT$ and $t < t'$.
  \begin{OneLiners}
  \item[(a)] If $I_{t'}$ contains $t$, then $t'$ is an ancestor of $t$ in $F$.
  \item[(b)] If $t'$ is not an ancestor of $t$ in $F$, then
    then any node $t''$ in the subtree rooted at $t'$ satisfies 
    $t'' > t$.
  \end{OneLiners}
\end{restatable}

\begin{proof}
  The first property follows by a simple induction, which we omit.
  For the second property, suppose for a contradiction that $t'' < t$,
  and let $t'' = t_0, t_1, \ldots, t_r = t'$ be the path from $t''$ to
  $t'$ in the forest $F$. Since $t' > t$, there is an $i$ for which
  $t \in [t_i, t_{i+1}]$. Since the interval $I_{t_{i+1}}$ contains
  $t_i$, it also contains $t$. But then, $t_{i+1}$ should be an ancestor
  of $t$ by property~(a), which gives a contradiction.
\end{proof}

We now divide pages and times into classes. For each class $c$, we will consider a sub-forest $F_c$ of $F$. 
We say that a page $p$ is of {\em class} $c$ if $w(p)$ lies in the range $[2^c, 2^{c+1})$. We say that a node $t$ in the charging forest $F$ is of class $c$ if the corresponding page $\pdag_t$ is of class $c$. Let $V^c$ be the vertices of class $c$ in $F$. Let $F^c$ be the minimal sub-graph of $F$ which preserves the connectivity between $V^c$ (as in $F$). So the leaves of $F^c$ belong to $V^c$, but there could be internal vertices belonging to other classes. We now show how to account for the cost incurred for the vertices in $V^c$. Let $A^\star_T(c)$ be the total weight of the stars in $\astar_T$ corresponding to pages of class $c$. We say that a node in $F^c$ is a {\em lone-child} if it is the only child of its parent. 

\begin{claim}
\label{cl:c1}
The total effective cost incurred during the leaf nodes in $F^c$ and the internal nodes of class $c$  in $F^c$ which are not lone-children is  $O(\astar_T(c))$. 
\end{claim}
\begin{proof}
The effective cost incurred during each time of class $c$ is a constant times $2^c$. 
Since the number of internal nodes which are not lone children is  bounded above by the number of leaf nodes, it is enough to bound the effective cost incurred at the leaf nodes. For a page $p$ of class $c$, let $F^c(p)$ be the leaf nodes $t$ in $F^c$ for which $\pdag_t = p$. Let the times in $F^c(p)$ in increasing order be $t_1, t_2, \ldots, t_k$. Note that $I_{t_i}$ does not contain $t_{i-1}$ for $i=2, \ldots, k$---otherwise $t_{i}$ will be an ancestor of $t_{i-1}$ 
(\Cref{lem:forest}).

\Cref{cl:case1alt'}
now implies that $\astar_T$ contains a star for page $p$ during 
$(t_{i-1}, t_i]$. Thus, the total effective cost incurred during $F^c(p)$ can be charged to the stars in $\astar_T$ corresponding to page $p$. Since all the leaf nodes in $F^c$ belong to class $c$, he result follows. 
\end{proof}

It remains to account for the times in $V^c$ which have only one child in $F^c$. 

\begin{claim}
\label{cl:c2}
Let $t_1$ and $t_2$ be two distinct times of class $c$ which are lone-child nodes in $F^c$. Let $t_1'$ and $t_2'$ be the parents of $t_1$ and $t_2$ respectively. Then the intervals $[t_1, t_1']$ and $[t_2, t_2']$ are internally disjoint. 
\end{claim}
\begin{proof}
Suppose not. Say $t_1 < t_2 \leq t_1'$. First assume $t_2' > t_1'$. Then $I_{t_2'}$ contains $t_1'$ and so $t_2'$ must be an ancestor of $t_1'$. If $t_1'$ is same as $t_2$, then the result follows easily, otherwise $t_1'$ is a descendant of $t_2$ (since $t_2'$ has only one child). But then $t_1' < t_2$, a contradiction. 

The other case happens when $t_2' < t_1'$. In this case $I_{t_1'}$ contains $t_2'$ (since it contains $t_1$ and $t_1 < t_2'$). If $t_1 = t_2'$, the result again follows trivially. Otherwise $t_2'$ is a descendant of $t_1$, a contradiction. 
\end{proof}

The above Claim along with
\Cref{cl:case1alt} and \Cref{cl:c1}
show that the total cost incurred by times of class $c$ can be charged to $w(A^\star_T)$. Thus, if there are $K$ different classes, we get $O(K)$ approximation. To convert this into $O(\log n)$ approximation, we observe the following refinement of Claim~\ref{cl:case1alt}. For a class $c$, times $t_1 < t_2$, let $A^\star_T(c,[t_1,t_2])$ be the stars of $A^\star_T[t_1, t_2]$ which are of class $c$. 

\begin{claim}
  \label{cl:case1altr}
  Suppose times $t_1, t_2 \in \fT$ are such that $t_1 < t_2$ and
  $I_{t_2}$ contains time $t_1$. Let $\pdag_{t_1}$ be of class $c$. Then the effective cost at time
  $t_1$ is at most at most 
  $$10 \left( \sum_{c'=c-\log n-3}^{c} w(\astar_T(c',[t_1,t_2])) + \sum_{c' > c} \frac{w(\astar_T(c',[t_1, t_2])}{2^{c'-c}} \right). $$
\end{claim}

\begin{proof}
\Cref{cl:case1alt} 
shows that $w(P^\star_T[t_1, t_2])$ is at least $\frac{w(\pdag_{t_1})}{2}.$ Let $P'$ be the pages of weight at most $w(\pdag)/4n$ in $P^\star_T[t_1, t_2]$---since the total weight of these pages is at most $w(\pdag)/4$. Thus the pages of class $c-\log n - 2$ and higher contribute at least half of $w(P^\star_T[t_1, t_2])$. Further, if $P^\star_T[t_1, t_2]$ contains a page $p$ of weight higher than $w(p)$, then we can just charge it $w(\pdag_{t_1})$ (and may not even charge to other pages in $P^\star_T[t_1, t_2])$. The desired result now follows from
\Cref{cl:case1alt}. 
\end{proof}

\Cref{cl:c1} and~\Cref{cl:c2} along with~\Cref{cl:case1altr} 
imply that
the total cost incurred during times of class $c$ is a constant times $$
\sum_{c'=c-\log n-3}^{c} w(\astar_T(c')) + \sum_{c' > c}
\frac{w(\astar_T(c'))}{2^{c'-c}}.$$ Summing over all classes
yields~\Cref{thm:round-on}.

\section{Offline Algorithm for the \wPwTw and \wPwTwP Problems}
\label{sec:overlap}

In 
\Cref{sec:solving-wpwtw-using}, 
we gave an online algorithm for \wPwTw and \wPwTwP using 
online (integer) solutions to~\eqref{eq:IP} and \eqref{eq:IPp}. 
We shall now prove the following offline version of
\Cref{thm:round-on}. 

\begin{theorem}
\label{thm:offline}
There is a polynomial time algorithm that converts an
  $\alpha$-approximate integral solution to~(\ref{eq:IPp}) into a
  solution for the \wPwTwP instance and has  approximation ratio of $O(\alpha)$.
As a consequence, there is a polynomial time algorithm that converts an
  $\alpha$-approximate integral solution to~(\ref{eq:IP}) into a
  solution for the \wPwTw instance and has  approximation ratio of $O(\alpha)$.  
\end{theorem}

As in \Cref{thm:round-on}, we will actually prove the above theorem 
for the moire restricted \wPwTw problem. This is sufficient for the 
more general \wPwTwP problem as well, by the same reduction as the one
we used in \Cref{thm:round-on}. Namely, the requests that are satisfied
by the integer solution to an \wPwTwP instance are used to create an 
(equivalent) instance of the \wPwTw problem, and then the above theorem 
for the \wPwTw problem is applied to this instance to derive a valid
solution for the original \wPwTwP instance. 

In the offline setting, we can assume that a request interval for a page $p$ does not contain another interval for the same page---otherwise we can always remove the outer interval. 
Let $\astar$ be an integral solution to~\eqref{eq:IP}, and we 
want
to convert it to a feasible solution to the underlying \wPwTw
instance. As discussed in
\Cref{sec:solving-wpwtw-using}, 
this will be done by 
adding a reverse delete step to
\Cref{algo:first} 
(which considered the special case when all the request intervals for a particular page were mutually disjoint).

\begin{figure}[t]
  \removelatexerror
  \begin{algorithm}[H]
    \ForEach{$t=0, 1, \ldots$}{
      \textbf{let} $I_t$ be the interval with deadline $t$, and let $p_t
      \gets \page(I_t)$ \\
      \label{l:firsttestf} \If{cache $C(t)$ full and $I_t$ not satisfied}{
        \textbf{evict} the least-weight page $p_{\min}$ in
        $C(t)$ \label{l:evict1f} \\
        \If{$w(p_t) \leq 2\,w(p_{\min})$}{

          \medskip
          $Z^\star \gets \emptyset$. \label{l:Zb1f} \\
          \For{every page $p$ in $C(t)$}{
            $I^p_t \gets$ the request interval $I$ with $\page(I) = p$ and
            largest ending time $e(I) < t$. \label{l:Ipeef} \\
            \textbf{if} $\dext{I^p_t,t}$ is hit by $\astar$ then add %
            $p$ to
            $Z^\star$. \label{l:Ze1f}
          }
          \label{l:Udf}
          $U \gets$ unsatisfied request intervals active at time $t$ %
          (one per page, page requests are disjoint%
          ).
          \\
          $U^\circ \gets \{I \in U \mid \exists t' \in I \text{ with }
          (\page(I),t') \in \astar \}$ be intervals in $U$ \textbf{not} hit by
          $\astar$ \label{l:Ucircf}\\

          \medskip
          \textbf{serve} and \textbf{evict} all requests in $U \setminus
          U^\circ$. \label{l:serve1f} \\ 

            \medskip

            \textbf{let} $U^\circ_{\leq w}$ and $Z_{\leq w}^*$ denote pages in $U^\circ$ and $Z^\star$
            respectively with weight at most $w$. \\
            \textbf{let} $\pstar$ be a page in $Z^\star$ such that $w(U^\circ_{\leq
              2w(\pstar)}) \leq 2 \cdot w(Z^\star_{\leq w(\pstar)})$. \label{l:goodpf}
            \\ 
            \textbf{evict} all pages in $Z^\star_{\leq
              w(\pstar)}$. \label{l:evict2f} \\
            \textbf{serve} and \textbf{evict} all requests in $U^\circ_{\leq 2w(\pstar)}$. \label{l:serve2f} \\
        }
      }
      \textbf{if} $I_t$ not satisfied \textbf{then}
      bring page $p_t$ into
      cache. 
      \label{l:retainf}
    }
    \medskip
    \hl{$\Sat \gets$ set of requests serviced in line~}\ref{l:serve1f}
    \\ \label{l:revf}
    \For{\hl{every page $p$}}{
      \hl{$T_p \gets$ set of times when request intervals in $\Sat$ for page
      $p$ are serviced (in line~}\ref{l:serve1f})       \label{l:revpf}
      \\
      \hl{$\Sat'_p \gets$ maximal disjoint collection of request intervals
      for
      page $p$ in $\Sat$.} \\
      \hl{$T_p' \gets$ set of times in $T_p$ closest (on either side) to
      the two end-points of intervals in $\Sat'_p$} \\
      \hl{\textbf{cancel} all movements of $p$ into cache at times in 
      $T_p \setminus T_p'$ (during line~}\ref{l:serve1f}).
   \label{l:lastf} \\
    }
    \caption{ConvertOffline$(\text{IP solution }\astar)$}\label{algo:full}
  \end{algorithm}
  \caption{Offline Algorithm to Service Request Intervals in the General Case}
  \label{fig:algo-full}
  \Description{Offline Algorithm}
\end{figure}

The algorithm is shown in 
\Cref{algo:full}. 
The first part of the
algorithm until 
\cref{l:retainf} is same as in \Cref{algo:first}. 
However, we cannot pay for all the evictions in
line~\ref{l:serve1f}. Therefore, we {\em remove} some of these evictions
in lines~\ref{l:revf}--\ref{l:lastf}.  We describe the details of this
process now. We use $\Sat$ to denote the set of requests serviced during
\cref{l:serve1f}. 
Let $\Sat_p$ be
the requests in $\Sat$ that correspond to $p$ (and $T_p$ be the time at which they are served), and $\Sat'_p$ be a
maximal collection of disjoint intervals in $\Sat_p$. Since each of the
intervals in $\Sat'_p$ is hit by a distinct element of $A^\star$, we can
pay for the service of $\Sat'_p$. We define $T_p'$ to be the time
instances in $T_p$ which are closest on each side to the end-points of
the intervals in $\Sat'_p$, and hence $|T_p'| \leq 4|\Sat'_p|$. It is
not difficult to show that each interval in $\Sat_p$ has non-empty
intersection with $T_p'$, and so it suffices to service $p$ only during
the times in $T_p'$. This is why the algorithm is correct, and services
all requests; we prove these facts formally below.

For the analysis, we again give some supporting claims to show that the
algorithm is well-defined, and then bound the cost. The proofs of
Claim~\ref{cl:0}--\ref{cl:Z}, and Lemma~\ref{lem:Z} remain unchanged. We
restate these here for sake of completeness.

\begin{claim}
  \label{cl:0f}
  Suppose a page $p$ is evicted from the cache at time $t_1$ but is in
  the cache at the end of time $t_2 > t_1$. Then there must exist a
  request interval $I$ for page $p$ with $t_1 < s(I) \leq e(I) \leq t_2$.
\end{claim}

\begin{claim}
  \label{cl:Zf}
  The set $Z^\star$ defined in lines~\ref{l:Zb1}--\ref{l:Ze1} is non-empty.
\end{claim}

\begin{lemma}
  \label{lem:Zf}
  There exists a page %
  $\pstar \in Z^\star$ such that 
  \[ w(U^\circ_{\leq 2w(\pstar)}) \leq 2\, w(Z^\star_{\leq
      2w(\pstar)}). \]
\end{lemma}

\Cref{lem:Zf} shows the existence of page $p^\star$ in
\cref{l:goodpf}. The proof of the following claim is same as that of
\Cref{clm:pt-evicted}.

\begin{claim}
  \label{clm:pt-evictedf}
  Let $I_t$ be unsatisfied at time $t$. If $w(p_t) \leq
  2w(p_{\min})$, then the page $p_t$ belongs to either $U\setminus
  U^\circ$ in line~\ref{l:serve1f} or to $U^\circ_{\leq 2w(\pstar)}$ in
  line~\ref{l:serve2f}, and is served and evicted. Else $p_t$ is served
  by line~\ref{l:retainf}, and remains in the cache.
\end{claim}

We now show that even after removing some of the services for a page $p$
in lines~\ref{l:revpf}--\ref{l:lastf}, the algorithm services all the
requests in $\Sat_p$. In the claim below, we use the notation in
lines~\ref{l:revpf}--\ref{l:lastf}.

\begin{claim}
\label{cl:rev}
Every request interval in $\Sat_p$ has non-empty intersection with $T_p'$. 
\end{claim}
\begin{proof}
Let $I$ be a request interval in $\Sat_p$. First assume that it lies in $\Sat'_p$, and let $t \in T_p$ be the time at which it is services in line~\ref{l:serve1f}. Since $t$ lies between $s(I)$ and $e(I)$, the closest time in $T_p$ to the right of $s(I)$, call it $t'$,  must lie between $s(I)$ and $t$. Since $t' \in T_p'$, the result follows. 

Now assume $I \notin \Sat'_p$. So there must be an interval $I' \in \Sat'_p$ which overlaps with $I$. Since any two requests for the same page are non-nested, $I'$ contains either $s(I)$ or $t(I)$. Suppose it contains $s(I)$ (the other case is similar). Then $s(I') < s(I) < e(I') < e(I)$. Let $t$ be the time at which $I$ is serviced in line~\ref{l:serve1f}. Say $t$ lies to the right of $e(I')$. Then the time in $T_p$ which is closest to $e(I')$ on the right side, call it $t'$, lies in $T_p'$. But $t' \in [e(I'), t]$ and so it lies in $I$. The case when $t$ is to the left of $e(I')$ is similar -- there will be a time in $T_p'$ which lies in the interval $[s(I), e(I')]$ and so belongs to $I$ as well. 
\end{proof}
The above claim proves that the algorithm services all the request intervals. 

We now analyze the cost incurred by the algorithm. The analysis is again very similar to that in Section~\ref{sec:costg}. To pay for evictions in line~\ref{l:evict1f}, we maintain the invariant that each page $p$ in the cache has at most $w(p)$ ``load'' on it; pages
outside the cache have zero load.  If we bring in $p_t$ and if $w(p_t)
\geq 2w(p_{\min})$, its load becomes the load of $p_{\min}$ plus the
cost of evicting $p_{\min}$; thus the total load on $p_t$ is at most its
weight, and $p_t$ remains in the cache, maintaining the invariant. At
the end of the algorithm, the total load over all pages is at most the
weight of the pages in the cache, which is at most the optimum
cost. This adds one to the competitive ratio.  Else if $w(p_t) <
2w(p_{\min})$, we evict at least one page in $Z^\star$ (by~\Cref{cl:Zf})
and can charge evicting $p_{\min}$ to the eviction of that page---which
we show below how to charge to $\astar$.

Now we consider the cost incurred during line~\ref{l:serve1f}. Because of lines~\ref{l:revf}--\ref{l:lastf}, we do not pay for all of these request intervals. 
Instead, for a particular page $p$, we serve at most $4 |\Sat_p'|$ requests for $p$ in this line. Since the requests in $\Sat'_p$ are disjoint and each of them is hit by $A^\star$, we can charge the service cost to the elements of $A^\star$. 

Finally, we charge  the eviction cost for lines
\ref{l:evict2f}-\ref{l:serve2f}. This cost  is $O(w(Z^\star_{\leq w(\pstar)}))$ by
our choice of $\pstar$ in line~\ref{l:goodpf}. Observe that for each page $p$ 
 in $Z^\star_{\leq w(\pstar)}$, the doubly-extended interval $\dext{I^p_t,t}$
is hit by an element $(p,t') \in \astar$, so we want to charge to this
element of $\astar$. Moreover, each page in $Z^\star_{\leq w(\pstar)}$ is at
least as heavy as $p_{\min}$, so any of these elements of $\astar$ can
pay to evict $p_{\min}$ (and its load).  We finally show that no element
of $\astar$ can be charged twice in this manner. The proof is identical to that of 
Claim~\ref{lem:twice}. 
\begin{lemma}
  \label{lem:twicef}
  No element in $\astar$ can be charged twice because of evictions in
  lines~\ref{l:evict2f}-\ref{l:serve2f}
\end{lemma}

This completes the proof of~\Cref{thm:offline}.

\newcommand{\previous}[2]{\tau^{#1}_{#2}}
\newcommand{\Dinterval}[2]{D^{#1}_{#2}}

\section{Solving the Integer Program \eqref{eq:IPp} for \wPwTwP}
\label{sec:first-lp}

We now give algorithms to solve the integer program~(\ref{eq:IPp}),
both in the offline and online settings. The main challenge in the
offline case is that the LP relaxation has an unbounded integrality
gap, so just relaxing the integrality constraints and then rounding
will not suffice. Instead, we write a compact IP that has a smaller
gap, and also has fewer constraints. Let us consider an example
problem, that of picking $n-k$ out of $n$ items. All items have unit
weight, so any feasible solution has cost at least $n-k$. If variable
$x_i \in [0,1]$ indicates that we should pick item $i$, we can write
an integer linear constraint for every choice $\calC$ of $k+1$ items
saying that $\sum_{i \in \calC} x_i \geq 1$. But the LP relaxation of
this IP admits the fractional solution where $x_i := \frac{1}{k+1}$
for all $i$, and hence total cost $\frac{n}{k+1} \ll n-k$, showing a
large integrality gap. However, replacing these $\binom{n}{k+1}$
linear constraints by the compact form
$\{\sum_i x_i \geq n-k, x \in [0,1]^n\}$ gives a formulation having no
integrality gap; we use analogous ideas to address both the challenges
above. In the online setting, we need to solve and round the resulting
LPs online, which will require us to refine the primal-dual algorithms
of Bansal et al.~\cite{BBN}.

We handle the constraints for the right and double extensions
separately, in \S\ref{sec:rext-IP-solve} and \S\ref{sec:dext-IP-solve}
respectively; this at most doubles the cost of the solution. In
both cases, we reduce to the following interval covering problem:
\begin{definition}[Tiled Interval Cover]
  In the \emph{tiled interval cover} problem (\intervalcover), for
  each page $p \in [n]$, we are given a collection $\calI_p$ of
  disjoint intervals that partition the entire timeline. All intervals
  in $\calI_p$ have the same weight $w(p)$.  The goal is to select a
  minimum-weight subset of intervals from $\calI := \cup_p \calI_p$
  such that for every time $t$, at least $n-k$ of these selected
  intervals contain $t$.
\end{definition}
The offline algorithms to solve \intervalcover will rely on
total-unimodularity, and the online ones will reduce to primal-dual
algorithms for the classical paging problem. The details of these
solutions to \intervalcover appear in~\Cref{sec:int-cover}.

\subsection{IP Solution for Right Extension Constraints}
\label{sec:rext-IP-solve}

In this section, the focus is only on the right-extension constraints,
i.e., %
the
following IP:
\begin{alignat}{2}
  \min_{x, y \text{ Boolean}} \quad \sum_{p,t} w(p) \,  x_{p,t} + & \sum_I \ell(I) \, y_I  \tag{IP-Rp} \label{eq:IPR1p}\\
  \sum_{I \in \calC} \min\bigg(1, y_I + \sum_{t' \in \rext{I,t}}
  x_{\page(I),t'} \bigg) &\geq 1
  & \quad &\forall t \, \forall \calC,   \tag{R1p} \label{eq:1pp} 
\end{alignat}
again $\calC$ consisting of $k+1$ requests, each for a distinct page,
and each starting before time $t$.
The discussion about getting a
compact IP above can be used to show that constraint~(\ref{eq:1pp}) is
equivalent (for integral solutions) to the
following constraint:
\begin{alignat}{2}
  \sum_{I \in \calC} \min\bigg(1, y_I + \sum_{t' \in \rext{I,t}} x_{\page(I),t'} \bigg) &\geq n-k
  & \quad &\forall t \, \forall \calC,   \tag{R2p} \label{eq:2pp} 
\end{alignat}
where $\calC$ now consists of $n$ requests, one for each of the pages,
and each starting before time $t$.

We now show how to approximately solve~(\ref{eq:IPR1p}) using an
algorithm for \intervalcover.  Consider an instance $\I$ of~(\ref{eq:IPR1p}). We assume that at time 0, there is a request
interval $[0,0]$ with infinite penalty for every page; this only
changes the optimum value by $\sum_p w(p)$. We now create an instance
$\I'$ of \intervalcover by creating a collection of intervals $\cK_p$
for each page $p$ using the procedure in~\Cref{fig:interval}; each of
these intervals will have weight $w(p)$. Essentially each such interval
is obtained by a minimal collection of original request intervals
corresponding to $p$ in $\I$ such that their total penalty exceeds
$w(p)$. \agnote{Figure fix.}

\begin{figure}[t]
  \centering
  \includegraphics[width=6in]{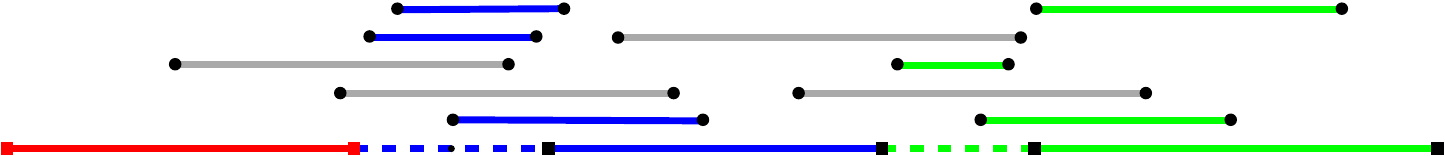}
  \caption{Constructing $\cK_p$ for page $p$: assume each request
    interval has the same penalty $\ell(I) = 1$, and $w(p) = 4$. The
    request intervals in gray intersect with the previous intervals in
    $cK_p$, so we increase $t$ until we see four non-intersecting
    intervals, and create the interval in $\cK_p$ based on the
    endpoint of this fourth interval (with endpoints  denoted by squares).}
  \label{fig:kp}
\end{figure}

\begin{figure}[h]%
    \begin{procedure}[H]
      Initialize $\cK_p  \gets \emptyset, t^\star \gets 0.$ \\
      \For{$t =1, 2, \ldots$}{ 
          ${\cal I}_p' \gets$ set of request intervals for $p$ which are contained in $[t^\star, t]$. \\
          \If{ the total penalty of the intervals in ${\cal I}_p'$ exceeds $w(p)$ }{
          Add $[t^\star,t)$ to $\cK_p$. 
           \\
          Update $t^\star \gets t$. } }
    \end{procedure}
  \caption{The online procedure to construct the partition $\cK_p$ for
    page $p$.}
  \label{fig:interval}
\end{figure}

\begin{lemma}[Forward Direction]
  \label{lem:forward-Rp}
  Consider an integral solution $(x,y)$ to $\I$. Then there is a
  solution $\calS$ to $\I'$ of cost at most
  $2 (\sum_{p,t} w(p) x_{p,t} + \sum_I \ell(I) y_I)$.
\end{lemma}

\begin{proof}
  For an interval $I'$ in $\cK_p$, let $\rt(I')$ to be the interval in
  $\cK_p$ which lies immediately to the right of $I'$. For every page
  $p$ and interval $I' \in \cK_p$, we add both $I'$ and $\rt(I')$ to the
  solution $\calS$ if either of these conditions is satisfied:
  (i)~$x_{p,t} =1$ for some time $t \in I'$, or (ii)~$y_{I}=1$ for
  every request interval $I \in \I$ for page $p$ such that the
  interval $I$ is contained within the interval $I'$.
  The cost guarantee is easy to see. We charge the cost of $I'$ and
  $\rt(I')$ to $w(p)x_{p,t}$ in case (i), and to the total penalty of
  request intervals for page $p$ contained within $I'$ in case (ii).

  Now to prove the feasibility of this solution $\calS$: fix a time
  $t$. For each page $p$, let $I'(p) \in \cK_p$ be the rightmost
  interval which ends before $t$. Classify the set of pages into two
  classes---let $P_1$ be the set of pages $p$ such that every request
  interval $I \in \I$ for $p$ which is contained within $I'(p)$ has
  $y_I=1$; define $I(p)$ to be any of these request intervals. Let
  $P_2 := [n] \setminus P_1$ be the remaining set of pages, i.e.,
  pages $p$ such that there is at least one request interval for it
  (call this request interval $I(p)$) contained within $I'(p)$ such
  that $y_{I(p)}=0$.

  Let $\calC$ be the collection of the pages $I(p)$ defined above,
  one for each page $p$. Applying constraint~\eqref{eq:2pp} to
  $\calC$, we get
  $$ |P_1| + \sum_{p \in P_2} \sum_{t' \in \rext{I(p),t}} \min(1, x_{p,t'}) \geq n-k. $$
  Hence, there is a set $P_2' \sse P_2$ of cardinality $n-k-|P_1|$
  such that for any page $p \in P_2'$, there is a time
  $t' \in I(p) \subseteq I'(p)$ with $x_{p,t'}=1$. It follows that we
  will pick intervals $I(p)$ and $\rt(I(p))$ into the solution
  $\calS$. Moreover, the collection $\calS$ contains the intervals
  $I(p)$ and $\rt(I(p))$ for each page $p \in P_1$. Since all these
  intervals $\rt(I(p))$ contain $t$, we have chosen
  $|P_1 \cup P_2'| \geq n-k$ intervals containing $t$, and hence  $\calS$ is a feasible solution to $\I'$.
\end{proof}

\begin{lemma}[Reverse Direction]
  \label{lem:reverse-Rp}
  Let $\calS$ be a integral solution to the instance $\I'$. Then there is a
  solution $(x,y)$ to $\I$ of cost at most $3w(\calS)$.
\end{lemma}

\begin{proof}
  The solution $(x,y)$ is as follows: for each page $p$ and interval
  $I'=[t'_1,t'_2] \in \cK_p \cap \calS$, (i)~set
  $x_{p,t_1'} = x_{p,t_2'} = 1$, and also (ii)~for every request interval
  $I=[t_1,t_2]$ for page $p$ having $t'_1 \leq t_1 \leq t_2 < t_2'$
  (i.e., intervals contained within such a chosen interval $I'$ and ending strictly earlier) set $y_{I} =
  1$. Since $t_2$ is strictly smaller than $t_2'$, the construction of
  $\cK_p$ ensures that the total penalty cost of such intervals is at
  most $w(p)$. The cost guarantee for $(x,y)$ follows immediately. It
  remains to show feasibility.

  Fix a time $t$. For each page $p$, let $I'(p) \in \cK_p$ be the
  interval containing $t$. By the feasibility of $\calS$, there is a set $P_1 \sse \calS$ of $n-k$
  pages such that their corresponding intervals $I'(p)$ contain
  $t$. Consider some constraint~\eqref{eq:1pp} corresponding to a set
  $\calC$ of requests for instance $\calI$. For each page $p$, let
  $I(p) \in\calC$ be the request interval for $p$. We know that
  $I(p)$ starts before time $t$, but it may end after time
  $t$. For a page $p \in P_1$, two cases arise: (i) $I(p)$ is
  strictly contained in $I'(p)$, or (ii) $\rext{I(p),t}$ contains one
  of the two end-points of $I'(p)$. In the first case we must have
  set $y_{I(p)}$ to 1, whereas in the second case there is a time
  $t'$ (which is one of the two end-points of $I'(p)$) in
  $\rext{I(p),t}$ such that $x_{p,t'}$ is $1$. Since $|P_1| = n-k$, we
  infer that $(x,y)$ satisfies~\Cref{eq:1pp}.
\end{proof}

Combining these with \Cref{lem:offline-intcover}, we see that there is
a 6-approximation offline algorithm for \othercoverp.

\subsubsection{Implementing the Solution Online}
\label{sec:onlineimpl}
For the \emph{online setting}, we can construct the set $\cK_p$
online, and also approximately solve the \intervalcover problem in an
online fashion using the algorithm from~\Cref{sec:int-cover}.
However, the translation given in \Cref{lem:reverse-Rp} cannot be
implemented online. Specifically, it sets the $x$ variables for start
and end times of the chosen intervals $I(p)$ and also $y$ variables
for all intervals strictly contained within $I(p)$, but (i)~setting
$x_{p,t}$ variables for times $t_1 < t$ and $y_I$ variables for past
intervals violates the past-preserving property required
in~\S\ref{sec:solving-wpwtw-using}, and (ii)~the online construction
of $\cK_p$ means the end time $t_2$ of the interval $I(p) \in \cK_p$
is not known at time $t$.

The first issue is easy to fix: if we change the proof of
\Cref{lem:reverse-Rp} so that when interval $I'(p) = [t_1',t_2']$ is
added to $\calS$ at some time $t$, we set $x_{p,t} = 1$ instead of
$x_{p,t_1'}$, and also we set $y_I = 1$ for interval $I$ that are
contained with $I'(p)$ only from now on. This change makes it
past-preserving, and maintains correctness. As for the secdon issue,
that the online algorithm may not know the right end-point $t_2'$ of
$I'(p) \in \cK_p$ at time $t$, and hence cannot add the star in the
future, there is at most one such ``to-be-set'' variable for each page
$p$; moreover, all intervals for page $p$ containing the current time
$t$ are already hit variables corresponding to past times, or by this
single to-be-set variable. This satisfies the \emph{sparsity} property
of \S\ref{sec:solving-wpwtw-using}. 

Hence, combining the above reductions with the algorithmic results
on solving \intervalcover in~\Cref{sec:int-cover}, and implementing these changes in the online
setting, we get:

\begin{lemma}[Right-Extension Algorithms]
  \label{lem:onlinennp}
  There is an online $O(\log k)$-competitive algorithm to solve the
  right-extension constraints with penalties~(\ref{eq:IPR1p}). The
  solution satisfies the monotonicity, past-preserving, and sparsity
  properties required in \S\ref{sec:solving-wpwtw-using}.  Finally,
  there is an offline 6-approximation algorithm for this problem.
\end{lemma}

\subsection{IP Solution for Double Extension Constraints}
\label{sec:dext-IP-solve}

We now want to solve the double-extension constraints~(\ref{eq:2p}):
\begin{alignat}{2}
  \min_{x, y \text{ Boolean}} \quad \sum_{p,t} w(p) \,  x_{p,t} + & \sum_I \ell(I) \, y_I  \tag{IP-D1p} \label{eq:IP1Dp}\\
  \sum_{I \in \calC} \min\bigg(1, y_I + \sum_{t' \in \dext{I,t}}
  x_{\page(I),t'} \bigg) &\geq 1 - y_{I_t}
  & \quad &\forall t \, \forall \calC,   \tag{D1p} \label{eq:2p1p} 
\end{alignat}
where $\calC$ consists of $k$ requests, each for a distinct page (not
equal to page $p_t$), and each request interval ending before time
$t$. As in the previous section, we can  replace~(\ref{eq:2p1p}) by the following
constraints and get exactly the same integer solutions.
\begin{alignat}{2}
  \sum_{I \in \calC} \min\bigg(1, y_I + \sum_{t' \in \dext{I,t}}
  x_{\page(I),t'} \bigg) &\geq (n-k)(1 - y_{I_t})
  & \quad &\forall t \, \forall \calC.   \tag{D2p} \label{eq:2p2p} 
\end{alignat}
Here the set $\calC$ consists of $n-1$ requests, one for each distinct
page different from page $p_t$,  where each of these request
intervals ends before time $t$. The variable $y_{I_t}$ has two
roles---the variable $y_{I_t}$ on the right denotes whether we need to
consider the constraints~\eqref{eq:2-againp} corresponding to time
$t$, whereas the ones on the left denote whether those requests were
satisfied, or if their penalty is paid instead. It turns out that we
can drop the occurrences of the $y_I$ variables on the left, and get
a polynomial-sized covering IP instead.

For each page $p$, let $\cK_p$ be as defined in~\Cref{sec:rext-IP-solve}, and
let $\cT_p$ denote the right end-points of the intervals in $\cK_p$.
For time $t$ and page $p$, define time $\previous{p}{t}$ as follows: let
$I_t=[t_1,t]$ be the request interval which ends at time $t$, and
$p_t$ denote $\page(I_t)$. If $I'(p)$ is the last interval in $\cK_p$
which ends before $t_1$, then $\previous{p}{t}$ is the right end-point of
this interval $I'(p)$. (Of course, the time $\previous{p}{t}$ lies in the set
$\cT_p$). Define the interval $\Dinterval{p}{t} := [\previous{p}{t}, t]$.

\begin{figure}[h]
  \centering
  \includegraphics[width=5in]{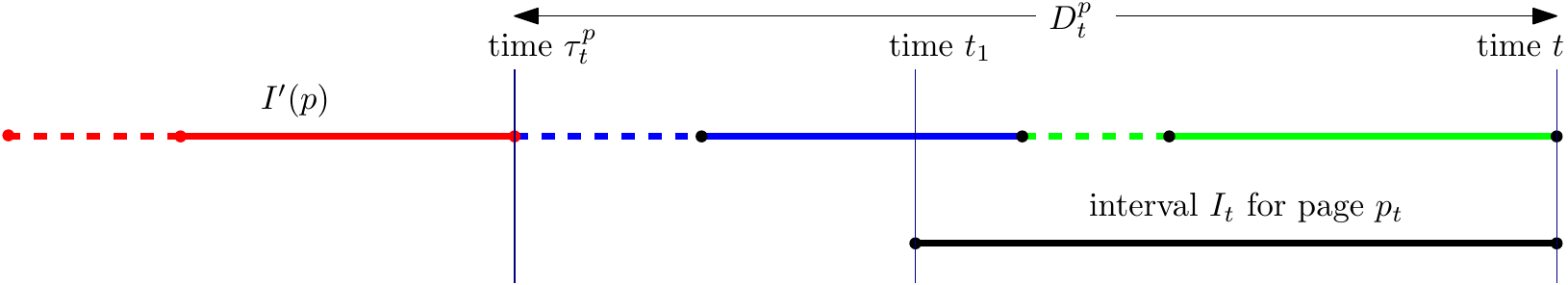}
  \caption{The definition of time $\previous{p}{t}$.}
  \label{fig:taupt}
\end{figure}

Consider the following compact IP, which is equivalent
to~(\ref{eq:IP1Dp}) up to constant factors, as we show next.
\begin{alignat}{2}
  \min_{x, y \text{ Boolean}} \quad \sum_{p,t} w(p) \,  x_{p,t} + & \sum_I \ell(I) \, y_I  \tag{IP-D3p} \label{eq:IP2Dp}\\
  \sum_{p: p \neq p_t} \min\Big(1, \sum_{t' \in \Dinterval{p}{t}} x_{p,t'}\Big) &\geq (n-k)(1-y_{I_t}) &\qquad&\forall t.
  \tag{D3p} \label{eq:2-againp}
\end{alignat}

\begin{lemma}
  \label{lem:twoips}
  Given an integer solution $(x,y)$ to~\eqref{eq:IP1Dp}, there is a solution
  $(x',y')$ to~\eqref{eq:IP2Dp} with cost at most three times as much.
  Conversely, if $(x',y')$ is an integer solution to~\eqref{eq:IP2Dp}, there is
  a corresponding solution to~\eqref{eq:IP1Dp} of at cost at most
  twice that of $(x',y').$
\end{lemma}

\begin{proof}
  For the first part, define $y'_I = y_I$ for all $I$, and
  $x'_{p,t}=x_{p,t}$ for all $p$ and $t$. For an interval
  $I \in \cK_p$, define $\rt(I)$ to be the interval immediately to its
  right in $\cK_p$. For each page $p$ and interval $I \in \cK_p$, we
  perform the following steps: let $t_1$ and $t_2$ be the right
  end-points of $I$ and $\rt(I)$ respectively. We set $x'_{p,t_1} =
  x'_{p,t_2} = 1$ if either of these conditions hold: (i)~there is a time
  $t \in I$ such that $x_{p,t}=1$, or (ii) $y_{I'}=1$ for every
  request interval $I'$ for $p$ which is contained within $I$. The cost
  guarantee for $(x',y')$ follows easily. 

  It remains to show that $(x', y')$ is feasible, consider a time $t$
  and assume $y'_{I_t}=y_{I_t} =0$ (otherwise~(\ref{eq:2-againp})
  follows immediately). Let $I_t=[t_1,t]$, and for page $p \neq p_t$,
  let $I'(p)$ be the last interval in $\cK_p$ ending before $t_1$, so
  that $\previous{p}{t}$ is the right end-point of $I'(p)$.  If there is a
  request interval $I$ for page $p$ which is contained in $I'(p)$ and
  $y_{I}=0$, we define $I(p)$ to be such an interval $I$, otherwise
  $I(p)$ is any request interval for $p$ contained within $I'(p).$ Let
  $P_1$ denote the pages for which the first case holds and $P_2$ be
  the second set of pages. Feasibility of~(\ref{eq:2p1p}) for
  these set of intervals implies that
  $$ |P_2| + \sum_{p \in P_1} \min\Big(1, \sum_{t' \in \dext{I(p),t}}x_{p,t'}\Big) \geq n-k.$$
  For pages $p \in P_2$, we would have set $x'_{p,\previous{p}{t}} =
  1$ due to condition~(ii). Furthermore, the interval $\dext{I(p),t}$ contains
  $\Dinterval{p}{t}$ and is contained in $I'(p) \cup \Dinterval{p}{t}$.  Now
  for page $p \in P_1$, suppose $t' \in \dext{I(p),t}$ is such that
  $x_{p,t'}=1$. Then either $t' \in \Dinterval{p}{t}$ in which case
  $x'_{p,t'} = x_{p,t'}$, or else 
  $t' \in I'(p)$ and we would have set $x'_{p.\previous{p}{t}} = 1$ by condition~(i).
  In either case, $\sum_{t' \in \Dinterval{p}{t}} x_{p,t'} \geq 1$. This shows
  that $(x',y')$ is a feasible solution to~\eqref{eq:IP1Dp}.

  We now show the converse. Let $(x',y')$ be a feasible solution
  to~(\ref{eq:IP2Dp}). We construct a feasible solution $(x,y)$ to
  (\ref{eq:IP1Dp}). As above, we first set $x=x', y=y'.$ Furthermore,
  for every $(p,t)$ such that $x'_{p,t} = 1$, let $I=[t_1,t_2]$ be the
  interval in $\cK_p$ containing $t$, and  set $x_{p,t_2} =
  1$. The cost guarantee for $x$ follows easily. To show feasibility,
  fix a time $t$ for which $y_{I_t} = 0$, and let $P$ be the set of
  pages for which the LHS equals 1 in
  constraint~\eqref{eq:2-againp}. For a page $p$, recall that $I'(p)$
  is the interval in $\cK_p$ which ended before $I_t$ started, but
  $\rt(I'(p))$ intersects $I_t$. Now consider the
  constraint~(\ref{eq:2p1p}) for time $t$ and set of pages $\calC$.
  We claim that the LHS term for every page $p \in P$ equals
  1. If $\dext{I,t} \supseteq \Dinterval{p}{t}$, then this claim follows
  from the fact that $p \in P$. So assume that $\dext{I,t} \not\supseteq \Dinterval{p}{t}$, and
  since both intervals share the same right endpoint, $\dext{I,t} \subseteq \Dinterval{p}{t}$. But then, by the greedy construction process, $\rt(I'(p))$ is contained in $\Dinterval{p}{t}$. 
  Since $p \in P$ we know that $x'_{p,t'}=1$ for some
  $t' \in \Dinterval{p}{t}$ and hence in $\rt(I'(p)) \cup I_t$. If $t' \in I_t$, then
  $t' \in \dext{I,t} $ where $I \in \calC$ is the request interval for
  $p$. Since $x_{p,t'}=1$ as well, we see that the LHS term for $p$
  in~\eqref{eq:2p1p} equals 1. On the other hand, if
  $t' \in \rt(I'(p))$, we will set $x_{p,.t''}$ to 1, where $t''$ is
  the right end-point of $\rt(I'(p))$. Since $t'' \in I_t$,
we have
  that $t'' \in \dext{I,t}$, and so the same conclusion holds.
\end{proof}

The compact LP can clearly be solved offline.  The following theorem,
whose proof is deferred the appendix, shows that the fractional
relaxation of~\eqref{eq:IP2Dp} can also be solved online losing only a
logarithmic factor.

\begin{restatable}[Solving the LP Online]{theorem}{solvedext}
  \label{thm:lponline}
  There is an $O(\log k)$-competitive online algorithm which maintains
  a solution to the fractional relaxation of~\eqref{eq:IP2Dp}.
At time $t$, the fractional solution
  only changes (and in fact, increases) the variables $x_{p,t}$ for
  all pages $p$, and $y_{I_t}$.
\end{restatable}

\begin{corollary}[Integral Penality Variables]
  \label{cor:lponline}
  Given an online fractional solution $(\tx, \ty)$
  to~\eqref{eq:IP2Dp}, we can maintain another online fractional
  solution $(\bx, \by)$ whose cost is at most twice of that of $(\tx,
  \ty)$, such that $\by_I$ is integral for every $I$. This solution
  also has the same property that at time $t$, it only increases the
  variables $x_{p,t}$ and $y_{I_t}$ corresponding to time $t$.
\end{corollary}

\begin{proof}
  Fix a time $t$ and the corresponding solution $(\tx, \ty)$. We set
  $\by_{I} = 1$ if $\ty_{I} > \nicefrac12$, and to 0 otherwise; we
  also  set $\bx_{p,t} = \min(1, 2\tx_{p,t})$. This
  at most doubles the cost of the solution, and  maintains
  feasibility.  Furthermore, at time $t$,
  $(\bx, \by)$ only increases the variables $\bx_{p,t}$ for all pages
  $p$ and $\by_{I_t}$.
\end{proof}

Let the problem defined by~\eqref{eq:IP2Dp} be called \pagecoverp, and
fix an instance $\I$ of this problem.  For rest of the discussion, we
maintain an online fractional solution $(\bx, \by)$ to $\I$ with the
properties mentioned in~\Cref{cor:lponline}.
In order to maintain an online integral solution to $\I$, we first
solve the problem for non-nested instances, and then prove an
``extension'' theorem to translate from non-nested instances to all instances.

\subsubsection{Solving Globally Non-Nested Cases of \pagecoverp}
\label{sec:solv-specialPC}

We say that interval $I = [t_1, t_2]$ is \emph{strictly nested} within
$I' = [t_1',t_2']$ if $t_1' \leq t_1 \leq t_2 \leq t_2'$, and either
the first or the last inequality is strict. (We drop the use of
strict, and simply say ``nested'' henceforth.) Two intervals are nested
if one of them is nested within another. Let $\cT'$ be some subset of
the timeline $[T]$ such that for every $t_1 \neq t_2 \in \cT'$, their
critical intervals $I_{t_1}, I_{t_2}$ are not nested. We define the
\specialpagecoverp problem, which solves the problem~(\ref{eq:IP2Dp})
where the constraints correspond to times $t \in \cT'$ and we are also
given that $\by_{I_t}=0$ for all $t \in \cT'$ (i.e., we are not paying
the penalty at these times). First, we show that the intervals
$\Dinterval{p}{t} = [\previous{p}{t},t]$ for any fixed page are also
non-nested.

\begin{claim}
  \label{cl:nonnest} For times $t, t' \in \T'$ and a page $p$ such
  that $p \not\in \{p_t, p_{t'}\}$, the intervals
  $\Dinterval{p}{t},\Dinterval{p}{t'}$ are non-nested.
\end{claim}

\begin{proof}
  Assume wlog that $t < t'$. Let $I_t=[t_1, t], I_{t'}=[t_1',t']$. By
  the non-nested property, $t_1 \leq t_1'$. Therefore, by construction,
  $\previous{p}{t} \leq \previous{p}{t'}$.
\end{proof}

We reduce the instance \specialpagecover to a tiled interval cover
problem with exclusions (\exintervalcover) instance. \exintervalcover
is like \intervalcover, where additionally for each time $t$ %
we are
specified a page $p_t$, and we cannot use the
intervals in $\calI_{p_t}$ for the coverage requirement at time
$t$. In~\Cref{sec:int-cover} we give a constant-factor approximation
algorithm and $O(\log k)$-competitive algorithm for \exintervalcover.

The reduction of the instance $(\I, \cT')$ of \specialpagecoverp to an instance $\I'$ of 
\exintervalcover proceeds as follows. For each page $p$, we build a disjoint 
collection of intervals $\calD_p$ as shown in \Cref{fig:dp}, by 
greedily picking a set of non-overlapping intervals in $\{D^p_t: t
\geq 0\}$ and extending them to partition the timeline. In the instance
$\I'$, the set $\I'(p)$ of intervals for page $p$ is given by the  
intervals in $\calD_p$, and the excluded page $p_t$ is the page
requested at time $t$. Furthermore, the cost of each interval in
$\I'(p)$ is given by $w(p)$, and the covering requirement $R$ at each time $t$ is $n-k$. 
The natural LP relaxation for $\I'$ has a variable  $z_I$ for each interval $I$ such that for every time $t$: 
\begin{gather}
  \sum_{I: I \in \I'(p), p \neq p_t, t \in I} z_I \geq n-k.
\end{gather}

\begin{figure}[h]%
    \begin{procedure}[H]
      Initialize $\calD_p  \gets \emptyset, t^\star \gets 0.$ \\
      \For{$t =1, 2, \ldots$}{ 
          \If{$D^p_t$ does not contain $t^\star$ in the interior}{
          Add $[t^\star, t]$ to $\calD_p$ \\
          Update $t^\star \gets t$. } }
    \end{procedure}
  \caption{The online procedure to construct the partition $\calD_p$ for
    page $p$.}
  \label{fig:dp}
\end{figure}

\begin{lemma}
  \label{lem:reductionp}
  Given the instance $(\I, \cT')$ as above, let $\I'$ be the corresponding
  \exintervalcover instance. 
  \begin{OneLiners}
  \item[(i)] Let $(\bx, \by)$ be a fractional solution to~\eqref{eq:IP2Dp} with $\by_{I_t} = 0$ for all $t \in \cT'$.
  Then there is a fractional solution to the above LP relaxation for $\I'$ of cost at most twice that of $(\bx, \by)$. 
  \item[(ii)] Let $\calS$ be an integral solution to $\I'$. Then there is an integral solution to the \specialpagecoverp instance $(\I, \cT')$ of cost at most twice that of $\cal S$. 
  \end{OneLiners}
 Moreover, both the above constructions can be done efficiently. 
\end{lemma}
\begin{proof}
Let $(\bx, \by)$ be a fractional solution to~\eqref{eq:IP2Dp}. We construct a fractional solution $z$ for the instance $\I'$ as follows. For every variable $\bx_{p,t}$, let $I$ be the interval in $\calD_p$ containing $t$, and let $\lt(I)$ be the interval in 
$\calD_p$ immediately to the left of $I.$ We raise both $z_I$ and $z_{\lt(I)}$ by 
$\bx_{p,t}$. Finally, if any $z_I$ variable exceeds one, we cap it at
one. The cost of $z$ is at most twice that of $\bx$. To show
feasibility, consider a time $t$. Let $I_1$ be the interval in
$\calD_p$ containing $t$ and $I_2 = \rt(I_1)$ be the interval immediately to the right of $I_1$ in $\calD_p$.
Both $I_1$ and $I_2$ contain intervals $D^p_{t_1}$ and $D^p_{t_2}$ for some times $t_1$ and $t_2$ respectively. Therefore,~\Cref{cl:nonnest} implies that $D^p_t$ is contained in $I_1 \cup I_2$. It follows that we will raise $z_{I_1}$ by at least 
$\min(1, \sum_{t' \in D^p_t} \bx_{p,t'}). $ From~\eqref{eq:2-againp} it follows that $z$ is a feasible solution. 

We now prove the second part of the lemma. Let $\calS$ be a feasible
solution to $\I'$. We will build an integral solution $(x, y)$, where
$y_{I_t}=0$ for all $t \in \cT'$, for $(\I, \cT')$. For every
$I=[t_1,t_2] \in \calS \sse \calD_p$, we set $x_{p,t_1} = x_{p,t_2} =
1$. 
The cost guarantee follows easily. To verify feasibility, consider a
time $t$. For each page $p \neq p_t$, let $I(p)$ be the interval in
$\I'(p)$ which contains $t$. Let $P(t)$ be the set of pages for which
$\calS$ contains $I(p)$. By feasibility of the $z$-solution, $|P(t)| \geq (n-k).$ The interval $D^p_t$ and $I(p)$ overlap (both of them contain time $t$). 
Also, $D^p_t$ cannot be contained in $I(p)$, by the way 
the set $\calD_p$ is constructed. Therefore, $D^p_t$ must contain one of the two end-points of $I(p) :=[t_1, t_2]$. Since we set both $x_{p,t_1}, x_{p,t_2}$ to 1, it follows that the LHS term corresponding to page $p$ in~\eqref{eq:2-againp} (for time $t$) is also 1. This shows that $(x,y)$ is feasible. 
\end{proof}

Combining \Cref{lem:reductionp} with the $2$-approximation for
\exintervalcover from \Cref{lem:offline-intcover} gives an 8-approximation algorithm for \specialpagecoverp.
As in
\Cref{sec:onlineimpl}, 
the reduction from the proof of
\Cref{lem:reductionp}(ii) 
can be carried out in an online manner. If the online algorithm for \exintervalcover selects an interval $I :=[t_1, t_2]$ at a time $t$ (note that $t \leq t_2$), we need to add the stars $(p,t_1)$ and $(p,t_2)$ in our solution for $\I$. It turns out that the proof of~\Cref{lem:reductionp} holds  if we add the stars $(p,t)$ and $(p,t_2)$ instead. Further, the star $(p,t_2)$ can be added at time $t_2$. 
 Since the set of constraints~\eqref{eq:2-againp} corresponding to
 time $t$ involve variables at $t$ and earlier only, the online
 algorithm need not {\em remember} at time $t$ the stars which will
 appear in future -- it can keep track of all the stars which have
 been added at time $t$, and any such star which corresponds to time
 $t' > t$ will only appear at time $t'$ in the algorithm. Thus, the
 algorithm satisfies the property that at any time $t$, it will only
 add stars corresponding to time $t$---we call such algorithms {\em
   present-restricted}; this is a stronger property than both
 past-preservation and sparsity (which were defined in \S\ref{sec:solving-wpwtw-using}).

\begin{lemma}
  \label{lem:onlinenonnest}
  There is an online $O(\log k)$-competitive present restricted algorithm to
 \specialpagecoverp. 
Moreover, there is an offline
  algorithm 8-approximation algorithm for \specialpagecover. 
\end{lemma}

\subsubsection{Algorithm for the General Case of \pagecover}
\label{sec:gen-pagecover}

We now consider the general setting where the critical intervals $I_t$
may be nested. \Cref{cor:lponline} shows that at every time $t$, we
know whether $\by_{I_t}=1$ or not, so we need only worry about times
for which $\by_{I_t} = 0$---call these times $\T$.  Let $\I$ be a
general instance of~\pagecoverp, where we want to obtained a cover for
times in $\T$. We show how to extend a solution for a
\specialpagecoverp sub-instance into one for the original instance
$\I$, while losing a constant factor in the cost. Let us give some
useful notation.
Given a set of times $\T$, a
subset $\calN$ is a \emph{non-nested net} of $\T$ if
\begin{OneLiners}
\item[(i)] for times $t_1 \neq t_2 \in \calN$, their critical intervals
  $I_{t_1}, I_{t_2}$ are non-nested, and 
\item[(ii)] for every time
  $t \in \T\setminus \calN$, there is a time $t' \in \calN$ such that
  $I_{t}$ contains $I_{t'}$.
\end{OneLiners}

A greedy algorithm to construct a a non-nested net $\calN$ of $\T$
simply scans times in $\T$ from left to right, and adds time $t$ to
$\calN$ whenever $I_t$ does not contain $I_{t'}$ for any
$t' \in \calN$. This procedure is implementable online: whenever we
see a time $t$, we know whether it gets added to $\calN$ or not. Given
a set $\T$ of times and a non-nested net $\calN$ of $\T$, we define a
map $\varphi: \T \setminus \calN \to \calN$ as follows---for a time
$t \in \T \setminus \calN$, let $\varphi(t)$ be the right-most time
$t' \in \calN$ such that $I_t$ contains $I_{t'}$.

\begin{claim}[Monotone Map]
  \label{cl:gamma}
  Let $\T$ be a set of times, $\calN$ be a non-nested net of $\T$ and $\varphi$ be the associated map as above. Then for any 
  $t_1', t_2' \in \T \setminus \calN$,
  $t_1' < t_2' \implies \varphi(t_1') \leq \varphi(t_2')$.
\end{claim}
\begin{proof}
  Suppose there are
  $t_1' < t_2' \in \T \setminus \calN$ such that
  $\varphi(t_1') = t_1 > t_2 = \varphi(t_2')$.  Let
  $I_{t_1} = [s_1, t_1]$ and $I_{t_2}=[s_2, t_2]$. Since these two
  intervals are non-nested, it must be the case that $s_1 \geq s_2$.
  But then $I_{t_2'}$ contains $I_{t_1}$ as well and we would have set
  $\varphi(t_2') = t_1$.
\end{proof}

Given an integer solution $(\bx,\by)$ for \pagecoverp, we identify it with a set
$\astar$ of stars, where $\astar := \{(p,t) \mid x_{p,t} = 1\}$.
For a time $t$ and a set of elements $\astar$, let $P(\astar, t)$ denote
the set of pages for which the corresponding intervals intervals $D^p_t$ are hit by $\astar$. I.e., we can rephrase
constraint~(\ref{eq:2-againp}) as wanting to find a set $\astar$ such that
$P(\astar, t) \setminus \{p_t\}$ has at least $n-k$ pages.  The main
technical ingredient is the following extension result:

\begin{restatable}[Extension Theorem]{theorem}{ExtDext}
  \label{thm:extension}
  There is an algorithm that takes a set $\T'$ of times, a non-nested net
  $\calN' \sse \T'$, the associated monotone map $\varphi$, and a set $\astar$, and outputs
  another set $B^\star \supseteq \astar$ such that
  \begin{OneLiners}
  \item[(i)] $P(B^\star, t) \supseteq P(\astar, \varphi(t))$ for all $t \in
    \T'\setminus\calN'$, and 
  \item[(ii)] $w(B^\star) \le 3\, w(\astar)$.
  \end{OneLiners}
  This algorithm can be implemented in online manner as well. More
  formally,  assume there is a present preserving online algorithm
  which generates the set $\astar_t$ at time $t \in \T$. Then there is
  a present preserving online algorithm which generates $B^\star_t$ at
  time $t \in \T$ and satisfies conditions~(i) and~(ii) above (with
  $\astar$ and $B^\star$ replaced  by $\astar_t$ and $B^\star_t$ respectively). 
\end{restatable}

We defer the proof to \Cref{sec:ext-thm}, and instead explain how
to use the result in the off-line setting first. We invoke the extension theorem twice. 
For the first
invocation, we use 
\Cref{thm:extension} 
with the entire set of times
$\T$,  a net $\calN$ and the associated monotone map $\varphi$, and with $\astar$ being a solution of
weight at most $8\, \opt(\I)$ given by 
\Cref{lem:onlinenonnest} 
on the
sub-instance $\calN$. This outputs a set $B^\star$ with
$w(B^\star) \leq 24\, \opt(\I)$. Moreover, since $A^\star$ is feasible
for $\calN'$, it follows from the first property of 
\Cref{thm:extension}
that
$|P(B^\star,t) \setminus \{p_t\}| \geq |P(\astar,\varphi(t))| - 1 \geq
n-k-1$ for every time $t \in \T \setminus \T'$.

For the second invocation, let $\T_1 \sse \T \setminus \calN$ be the
subset of times $t$ such that
$|P(B^\star, t) \setminus \{p_t\}| = n-k-1$, i.e., those with
unsatisfied demand. We use 
\Cref{thm:extension} 
again, this time with
$\T_1$, a net $\calN_1$ and the associated monotone map $\varphi_1$, and a solution
$\astar_1$ obtained by using 
\Cref{lem:onlinenonnest}
on the sub-instance
$\calN_1$. This gives us $B_1^\star$ with weight at most
$24\, \opt(\I)$. We output $B^\star \cup B_1^\star$ as our
solution. Somewhat surprisingly, this set $B_1^\star$ gives us the extra
coverage we want, as we show next.

\begin{lemma}[Feasibility]
  \label{lem:general}
  For any time $t$, $|P(B^\star \cup B_1^\star, t)| \setminus \{p_t\}| \geq n-k$. 
\end{lemma}

\begin{proof}
  We only need to worry about times $\T_1 \setminus \calN_1$. Consider
  such a time $t_1$. Let $t_2 := \varphi_1(t_1)$ and
  $t_3 := \varphi(t_2)$. For sake of brevity, let $p_i$ denote
  $p_{t_i}$, and $I_i$ denote $I_{t_i}$. Note that
  $I_3 \subset I_2 \subset I_1$, and since there are no nested intervals
  of the same page, also $p_1 \neq p_2 \neq p_3$.
  Recall: we want to show $|P(B^\star \cup B_1^\star, t_1)| \setminus \{p_1\}| \geq n-k$.
  
  Note that $P(B^\star, t_2)$ contains $P(A^\star, t_3)$, and the latter
  has size at least $n-k$ by construction. Hence, for $t_2$ to appear in
  $\T_1$, we must have $|P(B^\star, t_2)| = n-k$ and
  $p_2 \in P(A^\star, t_3)$. 
  Since $I_1$ contains $I_3$, $D^{p_2}_{t_1}$ contains
  $D^{p_2}_{t_3},$ and so, $B^\star$ hits $D^{p_2}_{t_1}$, i.e., $p_2 \in P(B^\star, t_1).$

  Since $P(\astar_1, t_2) \setminus \{p_2\}$ has size $n-k$ (by
  construction of $\astar_1$), and $P(B_1^\star, t_1)$ contains
  $P(\astar_1, t_2)$, it follows that
  $P(B_1^\star, t_1) \setminus \{p_2\}$ also has size at least $n-k$.
  Therefore, $P(B^\star \cup B_1^\star, t_1)$ has size at least $n-k+1$
  because $P(B^\star, t_1)$ contains $p_2$. This implies the lemma.
\end{proof}

Since $w(B^\star \cup B_1^\star) \leq 48 \opt(\I)$, we get a 48-approximation algorithm for \pagecoverp. It is easy to check that these arguments carry over to the online case as well; we briefly describe the main steps. The set $\calN$ can be generated in an online manner using a greedy algorithm (as mentioned in the beginning of this section). We invoke the online algorithm in
\Cref{lem:onlinenonnest} 
to get a present restricted solution $\astar_t$ for all $t \in \calN$.
\Cref{thm:extension} 
implies that the present restricted solution $B^\star_t$ can be constructed at time $t$. Given $B^\star_t$, we can tell whether a particular time $t$ qualifies for being in $\T_1$. The same argument can now be repeated to show that we can maintain $(B^\star_1)_t$ for all $t \in \T_1$. Combining this with
\Cref{lem:onlinenonnest}, 
we get 
\begin{lemma}
  \label{lem:onlinenest}
  There is an online $O(\log k)$-competitive present restricted algorithm to
 \pagecoverp. 
Moreover, there is an offline
  algorithm 48-approximation algorithm for \pagecoverp. 
\end{lemma}

\begin{corollary}
\label{cor:gen}
There is a constant factor (offline) approximation algorithm for~\eqref{eq:IPp}. Further, there is an $O(\log k)$-competitive online solution that satisfies the past-preserving and sparsity property as in~\S\ref{sec:solving-wpwtw-using}. 
\end{corollary}

\begin{proof}
Let $\I$ be an instance of \wPwTwP. 
Let $\astar_1$ and $\astar_2$ be (offline) solutions for~\eqref{eq:IPR1p} and~\eqref{eq:IP2Dp} as guaranteed by
\Cref{lem:onlinennp} and~\Cref{lem:onlinenest} 
respectively. 
\Cref{lem:twoips} shows that $A_2^\star$ can be mapped to a solution $A_3^\star$ that satisfies~\eqref{eq:IP1Dp}, 
and $\cost(A_3^\star) \leq 2~\cost(A_2^\star). $ Further $A^\star := A_1^\star \cup A_3^\star$ is a feasible solution to~\eqref{eq:IPp}. Since~\eqref{eq:IPR1p} and~\eqref{eq:IP1Dp} are special cases of~\eqref{eq:IPp},
\Cref{lem:onlinennp,lem:onlinenest,lem:twoips} imply that 
$$ \cost(A_1^\star) + \cost(A_3^\star) \leq 6~\opt(\I) + 2~\cost(A_2^\star) \leq 6~\opt (\I) + 96~\opt(\eqref{eq:IP1Dp})
\leq O(1) \cdot \opt(\I).$$

The online version follows analogously. 
Note that the conversion from $\astar_2$ to $\astar_3$ in~\Cref{lem:twoips} can be carried out in an online manner, and if $\astar_2$ is present restricted, then so is $\astar_3$. 
Since $\astar_1$ and $\astar_3$ are past preserving, so is $\astar.$ Also the sparsity of $\astar_1$ and the fact that $\astar_3$ does not add any star in the future implies that $\astar$ also satisfies the sparsity property. 
\end{proof}

\Cref{cor:gen}, along with~\Cref{thm:round-on} and~\Cref{thm:offline}, implies~\Cref{thm:main-on} and~\Cref{thm:main-off} respectively.
The integrality gap of~\eqref{eq:IPp} is constant for the following reason -- the integrality gap of the LP relaxations for \pagecover and \specialpagecover are $O(1)$, and the reductions in
\Cref{lem:forward-Rp,lem:reverse-Rp,lem:reductionp} 
also hold between the fractional solutions to the corresponding problems.

\section{Extension to Paging with Delay}
\label{sec:penalties}

In this section, we show a simple reduction from the {\em (weighted) Paging
  with Delays} (\wPwD) problem to the \wPwTwP problem.
which allows us to translate the
results of the previous sections giving an $O(1)$-approximate offline
algorithm and an $O(\log k \log n)$-competitive online algorithm for
the \wPwTwP problem to get the same asymptotic performance for the
\wPwD problem.

We transform an instance $\I$ of \wPwD to an instance $\I'$ of
\wPwTwP as follows. Recall that each request in $\I$ is specified by a triple
$(p,t, \penaltyfn)$, where $p$ is the requested page, $t$ is the time
at which this request is made, and
$\penaltyfn: \{t, t+1, \ldots, \} \to \R_{\geq 0}$ denotes the
non-decreasing loss function associated with it. We may assume
without loss of generality that $\penaltyfn(t) = 0$, since otherwise
we can work with the function
$\penaltyfn'(t') := \penaltyfn(t') - \penaltyfn(t)$, and the
competitive ratio is no worse. To model this request, we create an
ensemble of intervals $[t,t']$ for each $t' \geq t$ in the \wPwTwP
instance $\I'$, where the penalty for the interval $I := [t,t']$ is
$\penaltyfn(t'+1) - \penaltyfn(t')$, for each $t' \geq t$.

To see the equivalence, suppose this request $(p,t,\penaltyfn)$ is
served at time $t'$---i.e., the page $p$ enters the cache after time
$t$ only at time $t'$. Then all intervals in its ensemble ending at
later times are also satisfied. Moreover, intervals ending at earlier
times $t,t+1, \ldots, t'-1$ are not satisfied, and their penalty adds
up to $\penaltyfn(t') - \penaltyfn(t) = \penaltyfn(t')$, as desired.
Given this equivalence and the algorithmic results for the \wPwTwP
problem, we get:
\begin{theorem}
\label{thm:wPwD}
    There is an $O(\log k\log n)$-competitive online algorithm 
    and an $O(1)$-approximate offline algorithm for the \wPwD problem.
\end{theorem}

This completes the proof of \Cref{thm:main-on,thm:main-off}.

\subsection*{Acknowledgments} We thank Ravishankar Krishnaswamy for
valuable discussions about this problem; many of the ideas here
arose in discussions with him.  This research was done under the
auspices of the  Indo-US Virtual Networked Joint Center IUSSTF/JC-017/2017.
AG was supported in part by NSF award CCF-1907820. DP was supported
in part by NSF award CCF-1535972, and an NSF CAREER award CCF-1750140.

{\small
\bibliographystyle{alpha}
\bibliography{paper}
}

\appendix
\section{NP-Hardness of \wPwTw}
\label{sec:np-hard}

We now show that the \wPwTw is APX-hard, even when the cache size $k=1$
and we have unit weights. The reduction is the same as that of Nonner
and Souza~\cite{NonnerS09} for the joint-replenishment problem, and we
give it here for completeness. The reduction is from the (unweighted)
\VC problem on bounded-degree graphs.

Consider an instance $\I$, consisting of a graph $G=(V,E)$, of the \VC
problem. We reduce it to an instance $\I'$ of the \wPwTw problem. In the
instance $\I'$, we have one page $p_e$ for every edge $e \in E$. We also
have a special page $p^\star$. All pages have unit weight and the cache
size $k$ is 1. We now specify the request intervals for each page. The
timeline $T$ is the line $[0,|V|+1]$. For the page $p^\star$ we have
request intervals $[t,t]$ for every integer $t \in T$, i.e., this page
must be in the cache (or brought into the cache) at each integer time
$t$. Now consider the page $p_e$ for the edge $(u,v) \in E$. Assume wlog
that $u < v$. We have three request intervals for this page $e$: $I_e^1
= [0, u], I_e^2 = [u,v], I_e^3 = [v,n+1]$, where $n$ denotes $|V|.$ Note
that these are closed intervals. This completes the description of the
reduction. We first prove the easier direction (see Figure~\ref{fig:red} for 
an example).

\begin{figure}[ht]
    \centering
    \includegraphics[width=3in]{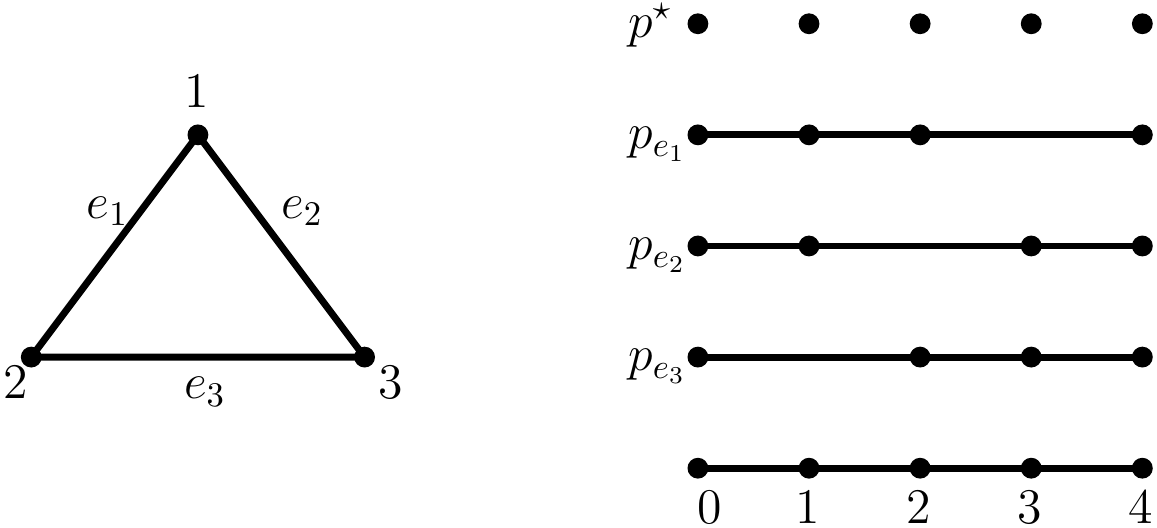}
    \caption{\emph{Illustration of the reduction from \VC to \wPwTw. The page $p^\star$ is 
    requested at each time. All the other pages have three request intervals, shown by solid lines with end-point delimiters. }}
    \label{fig:red}
\end{figure}

\begin{claim}
  \label{cl:np1}
  Suppose there is a vertex cover of $G$ of size at most $r$ (in the
  instance $\I$). Then there is a solution to $\I'$ of cost at most $r +
  2|E|+1.$
\end{claim}

\begin{proof}
  Let $U$ be a vertex cover of size $r$. The caching schedule is as
  follows: we will ensure that at the end of each time $t$, the page
  $p^\star$ is in the cache. This will ensure that all requests for
  $p^\star$ are satisfied.  For every $u \in U, $ we do the following:
  let $N(u)$ be the edges incident to $u$ in $G$. We bring each of the
  pages $p_e \in N(u)$ in the cache and the evict it. At the end of this
  process (at time $u$) we bring $p^\star$ back in the cache.

  For every edge $e \in E$, we have ensured that we bring $e$ in the
  cache either at time $u$ or $v$ (or both). If we bring in $e$ at both
  the times, we have satisfied all the requests for $p_e$. Otherwise, we
  would have satisfied two out of the three request for $p_e$, and the
  unsatisfied request would be either $I_e^1$ or $I_e^3$. Let $E_1$ be
  the set of edges $e$ for which the request $I_e^1$ is unsatisfied, and
  let $I_e^3$ be the set of edges $e$ for which $I_e^3$ is
  unsatisfied. At time 0, we bring in and then evict all pages in
  $I_e^1$. Then we bring in page $p^\star$. At time $n+1$, we evict
  $p^\star$ and bring in and then evict all the pages in $E_3$. This
  yields a feasible caching schedule. The total number of times
  $p^\star$ is evicted is at most $r+1$ (at each of the times in $U$,
  and maybe at time $n+1$). Every page $p_e$ is evicted exactly
  twice. This proves the claim.
\end{proof}

\begin{claim}
  \label{cl:np2}
  Suppose there is a solution to $\I'$ of cost at most $r+2|E|+1$. Then
  there is a vertex cover of $G$ of size at most $r+1$.
\end{claim}

\begin{proof}
  Let $\cal S$ be a solution to the caching problem. For an edge $e$,
  let $T_e$ be the timesteps when $e$ is brought into the cache. Since
  $p^\star$ must be present at the end of each integer time, $T_e$ must
  have non-empty intersection with each of the three request intervals
  for $e$.
  We now modify $\cal S$ to a solution ${\cal S}'$ which has the
  following property: (i) The total cost of ${\cal S}'$ is at most that
  of $\cal S$, (ii) for every edge $e=(u,v)$, the corresponding set
  $T'_e$ in ${\cal S}'$ has non-empty intersection with $\{u,v\}$.

  Initialize ${\cal S}'$ and $T'_e$ to $\cal S$ and $T_e$
  respectively. While there is an edge $e=(u,v)$ such that $T_e'$ does
  not contain $u$ or $v$, we do the following: $T'_e$ must contain a
  distinct time in each of the intervals $I_e^1, I_e^2, I_e^3$ -- let
  $t_1, t_2, t_3$ denote these three times. Note that $t_2$ must be in
  the interior of $I_e^2$. Assume wlog that $u < v$. Instead of bringing
  in $p_e$ at times $t_1$ and $t_2$ (and evicting them at these times),
  we will bring in $p_e$ at time $u$. This will save us a cost of 1 in
  the total eviction cost of $p_e$. However, it may happen that earlier
  $p^\star$ was not getting evicted at time $u$, and now we will need to
  evict it (and then bring it back into cache) at time $u$. Still, this
  will not increase the cost of the solution.

  Thus, we see that $\cup_e T'_e$ must contain a vertex cover $U$ of
  $G$. Since $p^\star$ must be getting evicted at each of these times,
  the total cost of ${\cal S}'$ (and hence, that of ${\cal S}$) is at
  least $2|E|+|U|$. This implies the claim.
\end{proof}

Using the above two claims, we show that the \wPwTw problem is APX-hard. 
\begin{lemma}
  \label{lem:hard}
  Let $G$ be a graph of maximum degree~4. Suppose there is an
  $(1+\eps)$-approximation for $\wPwTw$ problem. Then there is a
    $(1+9\eps)$-approximation for \VC on $G$.
\end{lemma}

\begin{proof}
  Let $\cal A$ be the $\alpha$-approximation algorithm for $\wPwTw$. The
  algorithm for \VC on $G$ is as follows: use the reduction described
  above to get an instance $\I'$ of $\wPwTw$. Run $\cal A$, and then use
  the proof of Claim~\ref{cl:np2} to get a vertex cover for $G$.

  Suppose $G$ has a vertex cover of size $r$. Since the maximum degree
  of $G$ is 4, we know that $|E| \leq 4r$. Now Claim~\ref{cl:np1}
  implies that $\I'$ has a solution of cost at most $2|E| + r+1$, and so
  $\cal A$ outputs a solution of cost at most
  $(1+\eps)(2|E| + r + 1) = 2|E| + 2\eps |E| + (1+\eps) (r+1) \leq
  2|E|+(1+9\eps) r + O(1)$. Claim~\ref{cl:np2} now shows that there is a
  vertex cover of size at most $(1+9\eps)r + O(1)$ in $G$.  This proves
  the lemma; the additive 1 can be ignored because we can take multiple
  copies of $G$ and make $r$ as large as we want.
\end{proof}

Finally, the fact that vertex cover is hard to approximate to within
$\approx 1.02$ on $4$-regular graphs~\cite{ChlebikC06} implies that 
\wPwTw is $\approx 1.002$-hard, and completes the proof.

\section{Some Illustrative Examples}
\label{sec:exs}

\subsection{Evictions at Endpoints are Insufficient}
\label{sec:endpoints-example}

It is easy to check that we cannot hope to service every interval $I$ at
either $s(I)$ or $t(I)$, which we can do for the unweighted case.
Indeed, consider the following input: suppose $k=1$ and there is a very
heavy page which is requested at each time, and so we need to have it in
the cache at every time. Now there are $n$ unit weight pages, but there
request intervals are $[0,n], [1, n+1], [2, n+2], \ldots.$ 

\begin{figure}[ht]
    \centering
    \includegraphics[width=3in]{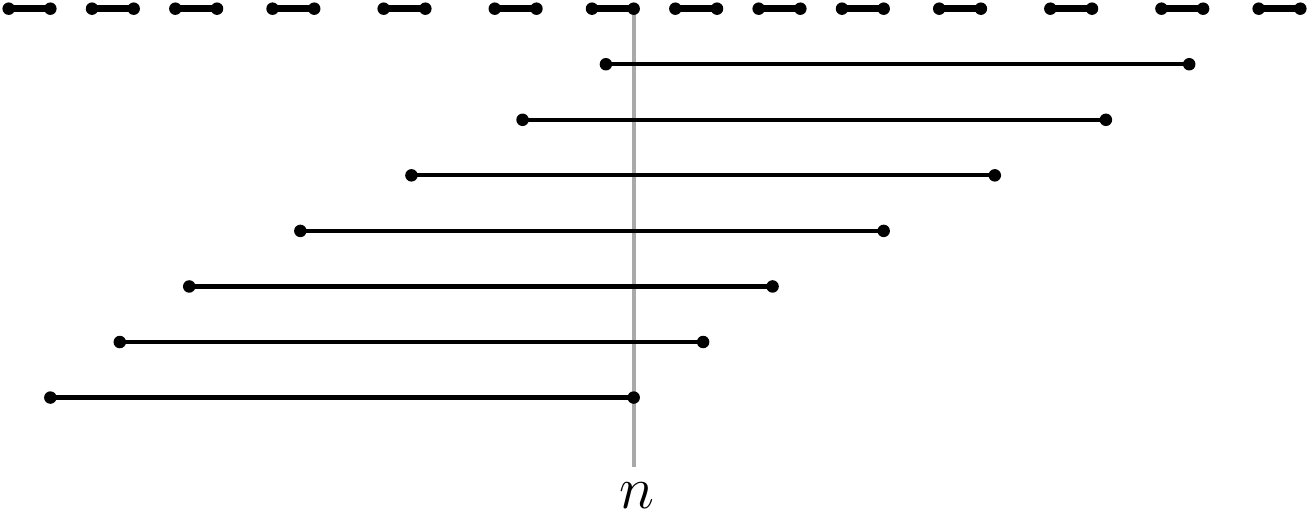}
    \label{fig:bad-example}
\end{figure}
The optimal solution is to service all these requests at time $n$,
because then we will evict the heavy page only once. Thus, our algorithm
needs to use these windows of opportunity to service as many cheap
requests as possible.

\subsection{An Integrality Gap for the Interval Hitting LP}
\label{sec:bad-ex}

We now consider a natural LP relaxation for \wPwTw which extends that
for weighed caching, and show that it has large integrality gap.  We
have variables $x_{p,J}$ for pages $p$ and intervals $J \sse [T]$,
indicating that $J$ is maximal interval during which the page $p$ is in
the cache for the entire interval $J$. Recall that we are allowed to
service many requests at each timestep, so each timestep may have up to
$n$ loads and $n$ evictions. To handle this situation, we ``expand'' the
timeline so that all such ``instantaneous'' services can be thought of
as loading each page in the cache for a tiny amount of time, and then
evicting it.  This will ensure that we can write a packing constraint in
the LP relaxation which says that no more than $k$ pages are in the
cache at any particular time.

Let $N$ be a large enough integer ($N \geq n$, where $n$ is the number
of distinct pages will suffice). We assume that all $s(I), t(I)$ values
for any request interval $I$ are multiples of $N$ (this can be easily
achieved by rescaling). Let $E$ denote the set of end-points of the
request intervals (so each element in $E$ is a multiple of $N$).  As
above, we have variables $x_{p,J}$, where the end-points of $J$ are
integers (which \emph{need not} be multiples of $N$).  The idea is that
between two consecutive intervals of $E$, we can pack $N$ distinct unit
size intervals, each of which may correspond to loading and then
evicting a distinct page. We can now write the LP relaxation:
\begin{alignat}{2}
  \min \sum_{p,J} w(p) \cdot & x_{p,J} & & \notag \\
  \label{eq:bad1}
  \sum_{J: J \cap I \neq \emptyset} x_{p,J} &\geq 1 &\qquad\qquad&
  \forall  \text{ request intervals $I$ with $p = \page(I)$} \\
  \label{eq:bad2}
  \sum_p \sum_{J : t \in J} x_{p,J} &\leq k &\qquad\qquad&
  \forall \text{ integer  times $t$} \\
  \notag
  x_{p,J} &\geq 0
\end{alignat}

\begin{theorem}
  \label{thm:bad-ex}
  The above LP has an integrality gap of $\Omega(k)$.
\end{theorem}

\begin{proof}
  Suppose we have $k$ ``heavy'' pages with weight $ k$ each, and $1$
  ``light'' page with weight 1. The request intervals for each of the
  pages are $E_0, E_1, \ldots, E_T$, where $E_i = [ikN, (i+1)kN]$, and
  $N$ and $T$ are suitable large parameters ($T, N > k^3$ will suffice).
  
  We first argue that any integral solution must have $\Omega(T)$
  cost. To see this, consider the request intervals
  $E_0, E_4, E_8, \ldots$ for the light page $p$.  The page $p$ must be
  brought at least once during each of these intervals -- say at
  timeslots $t_0, t_4, t_8, \ldots$, where $t_{4i} \in E_{4i}$ for all $i$.
  Notice that $E_{4i+2}$ lies strictly between $t_{4i}$ and $t_{4i+4}$, and hence
  each of the heavy pages must be present at least once during
  $E_{4i+2}$. Since there can be at most $k-1$ heavy pages in the cache
  at time $t_{4i}$, it follows that at least one heavy page must be
  brought into the cache during $[t_{4i}, t_{4i+4}]$. This argument
  shows that the cost of any integral solution must be $\Omega(Tk)$.
  
  Now we argue that there is a fractional solution to the LP of total cost $O(T)$. 
  For each heavy page $q$, we define $x_{q,J} = 1-1/k^2$, where $J = [0, (T+1)kN]$, i.e., $J$ is the entire timeline. 
  Notice  that each interval $E_i$ is of length $Nk$. We can therefore find $(k+1)k^2$ disjoint intervals of length 1
  each in it, and ``assign'' $k^2$ of these intervals to each of the
  $k+1$ pages. Let ${\cal S}_{i,p}$ be the 
  set of $k^2$ unit length intervals assigned to page $p$ (which could be the light page or one of the heavy pages). 
  For each heavy page $q$ and each unit length interval $H$ assigned to it, we set $x_{q,H}$ to $1/k^4$. For the 
  light page $p$ and each unit length interval $H$ assigned to it, we set $x_{p,H}$ to $1/k^2$. 
  
  Now we check  feasibility of this solution. Consider a heavy page $q$ and the request interval $E_i$ for
  it. The LHS of constraint~(\ref{eq:bad1}) for this request is $(1-1/k^2) + k^2/k^4=1$, where the first term corresponds
  to $x_{p,J}$ and the second term comes from the $k^2$ unit length intervals in ${\cal S}_{i,q}$. For the light page $p$ and the request interval $E_i$ for it, the LHS of this constraint is 1, because each of the $k^2$ unit length 
  intervals $H$ in ${\cal S}_{i,p}$ has $x_{p,H}$ equal to $1/k^2$. The constraint~(\ref{eq:bad2}) is easy to check -- for any time $t$, the LHS is at most  $k(1-1/k^2) + 1/k^2 \leq k$, because the first term comes from $x_{q,J}$ for each heavy page $q$, and the second term comes from the fact that all the unit length intervals are disjoint.  
  
  Let us now compute the cost of this solution. For a heavy page $q$, the total cost is $k(1-1/k^2) + T/k$, where the
  first term comes because of the long interval $J$ and the second term is because of the unit length intervals. 
  This is $O(T/k)$. Summing over all heavy pages, this cost is $O(T)$. For the light page, we have $Tk^2$ unit 
  length intervals, each to a fractional extent of $1/k^2$. Therefore the total cost here is $T$ as well. This proves the integrality gap of $\Omega(k)$. 
\end{proof}

The essential problem with this LP is that the heavy pages are being
almost completely fractionally assigned, leaving a tiny $\eps$ amount of
space. Since all the requests are long, they can be slowly satisfied
over $1/\eps$ time periods, which is much less cost than the cost of
actually evicting a heavy page.

\section{Proof of The Extension Theorem}
\label{sec:ext-thm}
\ExtDext*

\begin{proof}
  The procedure to obtain $B^\star$ from $A^\star$ is a simple greedy
  procedure, and appears in \Cref{fig:augment}: it goes over the times in $\T'\setminus \calN'$,
  and fixes any violations to the containment condition of the theorem by
  adding a new element to $B^\star$.
  \begin{figure}[h]
    \begin{procedure}[H]
      Initialize $\newS \gets \oldS$. \\
      \For{$t \in \T'\setminus \calN'$ in increasing order}{
        \For {every page $p \in P(\oldS, \varphi(t)) \setminus P(\newS,
          t)$ }{ \label{l:for}
          Add $(p,t)$ to $\newS$. }
      }
      \Return $\newS$.
    \end{procedure}
    \caption{The extension procedure to prove \Cref{thm:extension}}
    \label{fig:augment}
  \end{figure}
  For brevity, define $\T'' := \T'\setminus \calN'$. It immediately
  follows that for any $t' \in \T''$, the set $P(A^\star, \varphi(t))$
  is a subset of $P(B^\star, t)$.  We just need to bound the cost of
  $\newS$.

  Let the times in $\T''$ be $t_1 < t_2 < \ldots$, and \Cref{cl:gamma}
  shows that $\varphi(t_1) \leq \varphi(t_2) \leq \cdots$. The desired
  result now easily follows from the following claim.
  \begin{claim}
    \label{cl:B}
    Suppose we add $(p,t_{i}), (p, t_{j}), (p, t_{k})$ to the set
    $\newS$ for some page $p$ and times $t_{i} < t_{j} < t_{k}$.  Then
    there is a time $t \in [t_{i}, t_{k}]$ such that $(p,t) \in \oldS$.
  \end{claim}
  \begin{proof}
    Fix a page $p$ with $t_{i}, t_{j}, t_{k}$ as in the statement above,
    and let $t \leq t_{k}$ be the largest time such that
    $(p,t) \in \oldS$. %
    We now make a sequence of observations:
    \begin{itemize}
    \item[(i)] We claim that $\varphi(t_{k}) \geq t_{j}$ and
      $I_{\varphi(t_k)} \sse [t_j, \varphi(t_k)]$. If either is false then
      $I_{t_k}$ contains $t_j$, in which case there is no need to add
      $(p, t_k)$ to $\newS$, because $\newS$ already contains
      $(p,t_{j})$.
    \item[(ii)] Moreover, the interval $D^p_{t_j} \subseteq [t_i, t_j]$.
      Clearly $D^p_{t_j}$ ends at $t_j$ (by definition). If it
      starts before $t_i$, then there is no need to add $(p,t_j)$ to $\newS$.
    \item[(iii)] Next, $D^p_{ \varphi(t_k)} \subseteq [t_i, \varphi(t_k)]$: Since $\varphi(t_k) \geq t_j$, $D^p_{ \varphi(t_k)}$ starts after (or at the same time as) $D^p_{t_j}$ starts, and so this follows by~(ii) above. 
    \end{itemize}
    Now since $p$ lies in $P(\oldS, \varphi(t_k))$, statement~(iii) implies
    that we have $(p,t) \in \oldS$ for some $t \in [t_i, \varphi(t_k)]$.
  \end{proof}
  \Cref{cl:B} means that we can charge the three elements added to $\newS$ to
  this element $(p,t) \in \oldS$ that lies in between $[t_1, t_3]$. This
  proves the cost bound, and hence \Cref{thm:extension}.
  
  In the online setting, $B^\star_t$ can be easily constructed from $\astar_t$ using the procedure in~\Cref{fig:augment}; and it is easy to check it is also present restricted. 
\end{proof}

\section{The Tiled Interval Cover Problem}
\label{sec:int-cover}

In the \emph{tiled interval cover} problem (\intervalcover), the input
is the following.  For each page $p \in [n]$, we have a collection
$\calI_p$ of disjoint intervals that cover the entire timeline, with
each such interval having weight $w_p$. We also have a requirement $n-k$.
The goal is to pick some set of intervals from $\calI = \cup_p \calI_p$
that minimize their total weight, such that every time $t$ is covered by
$n-k$ different intervals. In the version \emph{with exclusions}
(\exintervalcover), the interval $E_t$ ending at time $t$ does not count
towards the requirement of $n-k$. (As always we assume that a unique
interval ends at each time.)

\subsection{The Offline Case}
\label{sec:intcover-offline}

\begin{lemma}
  \label{lem:offline-intcover}
  The linear relaxation for the \intervalcover problem is integral,
  whereas the that for the \exintervalcover problem has an integrality
  gap of at most $2$. 
\end{lemma}

\begin{proof}
  We can even show this for the case where the weights and requirements
  are non-uniform, i.e., each time $t$ has a potentially different
  requirement $R_t$, and each interval has a different weight $w_I$. Indeed, for the \intervalcover problem, the
  constraint matrix
  \begin{alignat}{2}
    \min_{z \in [0,1]^{|\calS|}} \quad\sum_{I} w_I \, & z_{I} \label{eq:IPforIC}\\
    \sum_{I \in \calS: t \in I} z_I & \geq R_t &\qquad & \forall t \notag
  \end{alignat}
  has the consecutive-ones property and forms a totally-unimodular
  system, so the linear relaxation has integer extreme points and an
  optimal integer solution can be found in polynomial time.

  Now let $z$ be a solution to the LP relaxation for (note the exclusion of 
  $E_t$ from the sum)
  \begin{alignat}{2}
    \min_{z \in [0,1]^{|\calS|}} \quad\sum_{I} w_I \, & z_{I} \label{eq:IPforICE}\\
    \sum_{I \in \calS: t \in I, I \neq E_t} z_I & \geq R_t &\qquad
    & \forall t.  \notag
  \end{alignat}
  Recall that $E_t$ is the interval ending at $t$ (though we will not
  need this for our solution). To construct
  an integer solution $\calS'$, first add to $\calS'$ all the intervals
  $I$ with $z_I \geq 1/2$. Now for each time $t$, let $R_t'$ be the
  residual coverage needed; i.e., define $R_t' := R_t - \#\{ I \in
  \calS' \mid t \in I, I \neq E_t\}$. Moreover, define $\tz_I := 2z_I$
  for $I \not\in \calS'$, and $z_I = 0$ for $I \in \calS'$. Clearly,
  $\sum_{I \ni t: I \neq E_t, I \not\in \calS'} \tz_I \geq 2R_t'$. Treat
  this as a solution to an \intervalcover instance on the subcollection
  $\calS \setminus \calS'$ (crucially, ignoring the exclusions) with
  these adjusted requirements $2R_t'$, and let $\calS''$ be an optimal
  integer solution. For each time $t$, there are now $(R_t-R_t')$
  non-excluded sets in $\calS'$, and at least $(2R_t' - 1)^+ \geq
  \max(R_t', 0)$ non-excluded sets from $\calS''$ covering it, which
  gives the desired coverage level of $R_t$. Due to the rounding up by a
  factor of $2$, the cost of the solution is at most $2\,w^\intercal z$.
\end{proof}

\subsection{The Online Case}
\label{sec:intcover-online}

The online model for \intervalcover and \exintervalcover is that
intervals are revealed online: specifically, the endpoint of an interval
is revealed only when it ends (and since we are dealing with tiled
instances, the next interval for that page starts immediately
thereafter).

In the online case, the \exintervalcover happens to be essentially
identical to formulation used in online primal-dual algorithms for
weighted paging, e.g., by~\cite{BansalBN12}. There are $n$ pages and a cache of
size $k$, so these constraints say that at time $t$, there must have
been $n-k$ pages apart from $p$ that are evicted since they were last
requested. Hence, the intervals for a page start just after each request
for the page, and end at the time of the next request. This means we can
simulate the end of intervals in $\calI_p$ by requesting page $p$. The
integer program is the following, where the page $I_t$ corresponds to
the page $p_t$ requested at time $t$.
\begin{gather*} \min\left\{ \sum_{p} \sum_{I \in \calI_p} w_p x_I \mid
    \sum_{I \in \calI \setminus \{I_t\}: t \in I} x_I \geq n-k
    ~~\forall t, \quad x_I \in \{0,1\} ~~ \forall I \right\}.
\end{gather*}
Using this connection and the result of~\cite{BansalBN12} immediately gives us
an $O(\log k)$ randomized online algorithm for \exintervalcover. To make
the online model closer to the rest of the paper, let us reformulate
the above IP as follows:
\begin{gather*} \min\left\{ \sum_{p,t} w_p x_{p,t} \mid \sum_{I \in
      \calI \setminus \{I_t\}: t \in I} \min\left(\sum_{t' \in I, t'
        \leq t} x_{p,t}, 1 \right) \geq n-k ~~\forall t, \quad x_{p,t}
    \in \{0,1\} ~~\forall p,t \right\}. 
\end{gather*}
It is easy to between these two formulations, using the correspondence
that at some time $t$, the variable $x_I$ has value equal to $\min( 1,
\sum_{t' \in I: t' \leq t} x_{pt'})$.
The algorithm from~\cite{BansalBN12} gives us an algorithm that only changes
the variables at the current time $t$; hence this is clearly a
\emph{past preserving} algorithm.

To get an algorithm for \intervalcover, we change the instance slightly:
we add in a new page $p_0$ (so there are $n+1$ pages) and make the cache
of size $k+1$. This new page has weight zero, so it can be brought in
and evicted at will. Now we request page $p_0$ immediately after a
request for any other page. (Denote the original request times by
integers, and the requests for $p_0$ by half-integers.) Observe that at
times $t - \nicefrac12$ when $p_0$ is requested, the constraints force
$(n+1)-(k+1) = n-k$ ``real'' intervals covering time $t$ to have been
chosen, which is precisely what we wanted. Now, suppose time $t$
corresponds to the interval in $\calI_{p_t}$ ending, and causing us to
request the page $p_t$. The paging constraint then asks for $n-k$ pages
except page $p_t$ to be chosen. But since page $p_0$ has zero weight, we
can choose it, so we need to only choose $n-k-1$ intervals from the rest
of the pages except $\{p_0, p_t\}$. Since $p_0$ will always be chosen,
this constraint is implied by the constraint at time $t -
\nicefrac12$. So this reduction to paging exactly models the
\intervalcover problem, and we get an $O(\log k)$-competitive algorithm
from~\cite{BansalBN12} again. We summarize the discussion of this section in
the following lemma:
\begin{lemma}
  \label{lem:bbn-tiled}
  There are randomized online algorithms for the \intervalcover and
  \exintervalcover problems that are $O(\log k)$-competitive against
  oblivious adversaries.
\end{lemma}

\section{Proof of~\Cref{thm:lponline}}
Recall the Integer Program~\eqref{eq:IP2Dp}:
\begin{alignat}{2}
  \min \sum_{p,t} w_p \, & x_{p,t} + \sum_I L(I) y_I \tag{IP-Dp} \label{eq:IP2Dp2}\\
  \sum_{p: p \neq p_t} \min(1, \sum_{t'=\cT(p,t)}^t x_{p,t'}) &\geq (n-k) (1-y_{I_t}) &\qquad&\forall t 
  \tag{D2p} \label{eq:2-againp2}
\end{alignat}

In this section, we prove the following result:
\solvedext*

For sake of brevity, we rename $y_{I_t}$ as $y_t$, $L_{I_t}$ as
$L_t$, $(n-k)$ as $R$, and the interval $[\cT(p,t),t]$ as $I(p,t)$: the only fact we use about this interval is that $I(p,t)$ always moves to the right~(\Cref{cl:nonnest}). We can now rewrite the linear relaxation of the above IP as 

\begin{alignat}{2}
  \min \sum_{p,t} w_p \, & x_{p,t} + \sum_t L_t y_t \tag{LP-p} \label{eq:LP2Dp2}\\
  \sum_{p: p \neq p_t} \min(1, \sum_{t' \in I(p,t)} x_{p,t'}) + R y_t &\geq R  &\qquad&\forall t \\
  x_{p,t}, y_t & \geq 0 & \qquad & \forall p,t
 \label{eq:lp2}
\end{alignat}

Observe that the cost of all $x_{p,t}$ variables corresponding to the same page $p$  is the
same. If the penalty costs $L_t = \infty$ we get the hard covering problem. The following algorithm is a simple extension of a result of
Bansal et al.~\cite{BBN}; we give it here for sake of 
completeness.

  All the variables are initialized to 0. 
  For an interval $I$ and page $p$, let $x_{p,I}$ denote
  $\sum_{t \in I} x_{p,t}$. Let $\delta = \frac{1}{k+1}$. The algorithm is
  simple: at each time $t$, if the corresponding constraint for  time $t$ is violated, then we
  raise some variables. Imagine this happening via a continuous
  process, with a clock starting at $\tau = 0$ and continuously increasing until
  the constraint is satisfied. Let
  $P_\tau = \{ p  \mid p \neq p_t, x_{p,I(p,t)} < 1\}$ be the pages in this
  constraint that are ``active'' at clock value $\tau$, i.e., the variables $x_{p,I(p,t)}$  are not
  already at their maximum value. We must have $|P_\tau| \geq k+1$,
  else the LHS of the constraint would have $R=n -k$ of the
  $x_{p,I(p,t)}$ values already at $1$, and the constraint would be
  satisfied. %
  Now raise the variables $x_{p,t}$ for every page $p$ in $P_\tau$ at the following rate:
  \[ \frac{dx_{p,t}}{d\tau} = \frac{x_{p,I(p,t)} + \delta}{w_p}. \]
  Also,
    \[ \frac{dy_t}{d\tau} = \frac{y_t R + \delta(|P_\tau|-k)}{L_t}. \]
  Note that we raise only the last variable $x_{p,t}$ for each interval $I(p,t)$, but
  it is raised proportional to the value of the entire interval. As
  these values rise, more pages fall out of the set $P_\tau$ until
  the constraint is satisfied.

  To show the competitiveness, let $x^*$ denote the optimal integer
  solution to~\eqref{eq:IP2Dp} after satisfying the constraint for time
  $t$. 
 Also, the  interval $I_{p,t}$ is not defined  for $p=p_t$, we define it for the sake of analysis to be same as $I(p,t')$ where 
$t'<t$ is
   the most recent time such that $p \neq p_{t'}$.
Now let the potential
  be
  \[ \Phi_t := 3 \sum_{p :  x^*_{p, I(p,t)} \geq 1} w_p \log
    \left( \frac{1+\delta}{\min(1, x_{p,I(p,t)}) + \delta} \right)
    + 3\; \sum_{s \leq t: y^*_s = 1} L_s \log \left( \frac{1+\delta}{y_s + \delta} \right)
 . \]
  Note that each term in the potential is non-negative. We show that
  the amortized cost of the algorithm with respect to this potential
  can be paid for by the optimal cost times $O(\log (1+1/\delta))$.

  First, suppose the constraint at time $t$ is revealed. This causes
  the current intervals $I(p,t)$ to possibly change, and hence some
  terms from the potential may disappear (since current intervals only
  move to the right). But dropping terms can only decrease the
  potential function. Next, let $OPT$ augment its
  solution. Suppose OPT decides to set $y^*_t = 1$, then OPT's cost is
  $L_t$, whereas the potential goes up by $3L_t \log
  (\frac{1+\delta}{\delta})$. Moreover, 
  for each variable $x^*_{p,t}$ that is set to $1$ (there
  is no reason to raise any other variable), the cost to $OPT$ is
  $w_p$, whereas the potential increase is at most
  $3 w_p \log (\frac{1+\delta}{\delta})$. Hence we have
  \[ \Delta \Phi \leq \Delta OPT \cdot 3 \log (1+1/\delta). \]
  Observe that since the current intervals $I(p,t)$ only move rightwards, we
  charge each optimal variable only once.

  Finally, the algorithm moves via the continuous process above. The
  instantaneous cost incurred by the algorithm is
  \begin{align}
   \frac{dALG}{d\tau} &= \sum_{p \in P_\tau} w_i
                        \frac{dx_{p,t}}{d\tau} + L_t \frac{dy_t}{d\tau} \\
                      &= \sum_{i \in P_\tau} (x_{p, I(p,t)} + 
                        \delta) + R \, y_t + (|P_\tau| -k)\delta \\
                      &= \sum_{p \in P_\tau} x_{p,I(p,t)} + R y_t + 
                        \delta\;(|P_\tau| + |P_\tau| - k) \\
                      &< R - (R - |P_\tau|) +
                        \delta\;(|P_\tau| + |P_\tau| - k) \\
                      &< (|P_\tau| - k) + 2|P_\tau|\,\delta.
  \end{align}
  But since $|P_\tau| \geq k+1$ and $\delta = \frac{1}{k+1}$, we get $|P_{\tau}|\,\delta \leq
  |P_{\tau}|-k$; this is the first time we use the value of
  $\delta$. Hence
  \[ \frac{dALG}{d\tau} \leq 3(|P_\tau| - k). \]
  Finally, using the chain rule
  and the definition of the continuous process, the decrease in potential is:
  \begin{align*}
   - \frac{d\Phi}{d\tau} &= 3 \sum_{p \in P_\tau: x^*_{p,I(p,t)} \geq 1}
    \frac{w_p}{x_{p,I(p,t)} + \delta} \cdot \frac{x_{p,I(p,t)} + \delta}{w_p}
    + 3\;\frac{L_t}{y_t + \delta} \cdot \frac{R y_t +
                           (|P_\tau|-k)\delta}{L_t} \cdot \mathbf{1}_{y^*_t = 1}  \\
    &\geq 3\cdot \#\{ p \in P_\tau \mid x^*_{p,I(p,t)} \geq 1\} +
      3(|P_\tau| - k) \cdot \mathbf{1}_{y^*_t = 1} %
  \end{align*}
  The first equality above uses the fact that for pages in $P_\tau$, $x_{p,I(p,t)}$ is strictly less
  than $1$, and so the truncation by 1 does not have any effect. Now either $y^*_t = 1$, in which
    case the second term gives us $3(|P_\tau| - k)$. Or else
    $y^*_t = 0$ and the second term is not present, but then $x^*$ is
    a feasible solution to the covering constraint~\eqref{eq:lp2} at time $t$. Therefore,  at most $n-R=k$ 
    intervals $I(p,t)$ are not
    hit by $x^*$. So the contribution of the first term in this case
    is at least $3(|P_\tau| - k)$. Putting these together, we get
  that $- \frac{d\Phi}{d\tau} \geq 3(|P_\tau| - k)$, and hence
  \[ \frac{d}{d\tau}(ALG + \Phi) \leq 0. \] This shows
  $\log (1+1/\delta)$-competitiveness. Using the setting of
  $\delta = \frac{1}{k+1}$ completes the proof of~\Cref{thm:lponline}.

\end{document}